\documentclass[11pt,a4paper]{article}

\usepackage[utf8]{inputenc}
\usepackage[T1]{fontenc}
\usepackage{lmodern}

\usepackage{amsmath,amssymb,amsfonts}
\usepackage{graphicx}
\usepackage{booktabs}
\usepackage{multirow}
\usepackage{siunitx}
\usepackage{caption}
\usepackage{subcaption}
\usepackage{float}
\usepackage{geometry}
\usepackage{setspace}
\usepackage{hyperref}
\usepackage{enumitem}
\usepackage{amsthm}
\usepackage[numbers]{natbib}
\usepackage{orcidlink}
\usepackage{authblk}

\newcommand{\cF}{\mathcal{F}}
\newcommand{\cN}{\mathcal{N}}
\newcommand{\cE}{\mathcal{E}}
\newcommand{\cG}{\mathcal{G}}
\newcommand{\tstar}{\theta^{\star}}
\newcommand{\that}{\hat{\theta}}

\newcommand{\Pstar}{\mathbb{P}^{\tstar}}
\newcommand{\Estar}{\mathbb{E}^{\tstar}}
\newcommand{\Et}{\mathbb{E}^{\star}_{t}}
\newcommand{\Emi}{\mathbb{E}^{\star}_{-i,t+1}}
\newcommand{\E}{\mathbb{E}}
\newcommand{\bone}{\mathbf{1}}
\newcommand{\R}{\mathbb{R}}
\newcommand{\abar}{\bar{a}}
\newcommand{\aund}{\underline{a}}
\newcommand{\norminf}[1]{\lVert #1 \rVert_{\infty}}

\theoremstyle{plain}
\newtheorem{theorem}{Theorem}
\newtheorem{lemma}{Lemma}
\newtheorem{proposition}{Proposition}
\newtheorem{corollary}{Corollary}
\theoremstyle{definition}
\newtheorem{hypothesis}{Hypothesis}
\newtheorem{definition}{Definition}
\newtheorem{remark}{Remark}

\hypersetup{
    colorlinks=true,
    linkcolor=blue,
    citecolor=blue,
    urlcolor=blue
}

\title{\textbf{Non-Bayesian Social Learning with Heterogeneous and Time-Varying Aggregation Rules}}

\author[1]{Sammy Oukaci\thanks{Contact: \href{mailto:o.sammy431@gmail.com}{\texttt{o.sammy431@gmail.com}}}}
\author[2]{Pol Rodríguez Farrés\, \orcidlink{0009-0005-3877-9720}\thanks{Contact: \href{mailto:pol.rodriguez@iiia.csic.es}{\texttt{pol.rodriguez@iiia.csic.es}}}}
\author[3]{Jesus Cerquides}

\affil[1]{ENSTA, France}
\affil[2]{IIIA-CSIC, Bellaterra, Spain}
\affil[3]{CSIC, Madrid, Spain}

\date{\today}

\begin{document}

\maketitle

\begin{abstract}
People constantly form opinions by combining what they observe with what friends, collegues, and communities tell them. Following this idea, this article generalizes earlier results in Non-Bayesian social learning on strongly connected social networks, where agents repeatedly update their beliefs by combining their own private signals with the beliefs of their neighbors. Previous work assumed that every agent aggregates information in the same way and often does not change in time, assuming agents never change how they reason and listen to others. We show that agents' beliefs still converge almost surely to the true state when agents use heterogeneous and time-varying (non-stochastic) aggregation rules under certain conditions. That is, even when people change their way to aggregate others opinions, the community can still reach the truth as long as information flows among its members. We also provide convergence rate results under this scenario and show that, for some of these functions, it is at least equal to the best convergence rate we have found in the literature. We present the widely used aggregation functions in the Constant Elasticity of Substitution (CES) family as a specific case of our general model, which include rules such as the log-linear, DeGroot, and the harmonic rules, and provide a simple illustrative example to show the capabilities of the model.
\end{abstract}

\noindent
\textbf{Keywords:} Non-Bayesian social learning, convergence rate, aggregation function, time-varying networks.

\section{Introduction}

In many decision-making situations, a group of people or devices (called agents) are trying to learn the same unknown fact or state about the world. This fact cannot be observed directly and, instead, each agent receives noisy hints or evidence about it. At the same time, agents share what they know with their neighbors and colleagues. This setup appears in sensor networks, distributed detection and control, robotics, and, more recently, in federated and decentralized learning. It also describes how communities and organizations form opinions, as when consumers learn about the quality of a product, investors about the value of an asset, or citizens about the benefits of a policy, by combining their private observations with what they hear from people they know. A central question in these settings is whether a group of agents can collectively learn the true state through repeated interaction, even when no single member has enough information to do so alone. Moreover, the speed of this process matters. If the group learns slowly, it keeps making decisions based on inaccurate beliefs, and the cost of those mistakes grows the longer it takes to correct them.

There are two general approaches to modeling the update and exchange of information and opinions, also called beliefs in such settings. In a fully Bayesian model, each agent keeps track of the complete history of her own observations and of the information received from other agents, and updates her beliefs using Bayes' rule. Although this procedure is statistically appropriate for a fully rational agent, it becomes difficult to implement as the network grows, because the amount of information that must be retained increases rapidly. Non-Bayesian social learning (NBSL) avoids these difficulties by limiting the amount of information that each agent retains. This is also a more realistic description of individuals, who rarely keep track of who told what to whom, and who typically act on the current opinions of those around them. In the standard formulation, agents combine only the last-period beliefs of their neighbors (also called imperfect recall) through an aggregation rule and then incorporate their own newly observed signal through a Bayesian update. Despite its simplicity, this approach can still lead to correct learning under suitable conditions.

The area of social learning builds on a line of work that begins with the linear averaging rule of DeGroot \cite{degroot1974reaching}, in which an agent updates her belief as a weighted average of the beliefs of her neighbors in a graph. The more modern formulation of non-Bayesian social learning, in which agents aggregate their neighbors' current beliefs through a log-linear rule and refine them with new private information, was introduced in \cite{jadbabaie2012non} and given axiomatic foundations in \cite{molavi2017foundations}, later published as \cite{molavi2018theory}. That work shows that log-linear aggregation of beliefs follows from imperfect recall together with additional restrictions on how agents process neighbors' information, and it characterizes the environments in which different classes of learning rules, including DeGroot, successfully aggregate information.

In fact, most non-Bayesian social learning models consider time-invariant aggregation rules, and often assume homogeneity across agents. This assumption is restrictive, since agents may vary the way they evaluate and aggregate information, as well as the weight they give to others. In social environments, people differ in how much and in what way they defer to their peers, some following the most convinced voices around them and others seeking a compromise among all views. The credibility of each source also changes as reputations are built or lost, as individuals gain experience, or as the institutional setting changes, for instance when a group moves from independent work to formal deliberation. More generally, the relative importance assigned to other agents may change throughout the learning process. The main contribution of this work lies in generalizing the results of \cite{molavi2018theory} (specifically the results obtained in their Theorems 4, 5 and 7), by considering time-varying (non-stochastic) and heterogeneous aggregation rules, which makes it possible to model complex and changing interactions between agents. When analyzing such a setting, two main questions arise. The first is whether agents continue to learn the true state i.e., whether all agents' beliefs converge to the true state while rejecting the alternatives.

One important caveat of our model is that, with agents represented in a network, the communication graph remains fixed. In other words, agents may change the weights assigned to their neighbors over time, but the set of neighbors does not change. This corresponds to situations in which social ties are stable over the horizon of learning, such as colleagues within an organization or members of a committee, while the credibility granted to each tie evolves. In more detail, the weights in the graph, i.e. the attention paid to neighbors, are actually defined in terms of the aggregation rule. This, together with the collective informativeness of the observed signals among agents, determines the conditions under which long-run learning and the convergence of beliefs to the true state occurs.

We then provide in this work the conditions under which the agents' beliefs will converge, almost surely, on the true state, when using time-varying (non-stochastic) and heterogeneous aggregation rules. Apart from this, we prove that any differentiable admissible aggregation rule under those conditions can be represented exactly as a log-linear rule with stochastic, time-dependent weights. This representation allows a common description of a broad class of aggregation rules and forms the basis for the analysis of learning rates. We then complement this result by providing a general bound (one obtained under the main assumptions) on the convergence rate. This bound depends only on the informativeness capacity of each agent, and not on the aggregation rule used. Therefore, no form of social interaction, regardless of how well it is structured, can enable a group to learn faster than what the information available to its members permits. Communication can share and reorganize existing information, but it cannot generate new information. Moreover, we tighten this bound by considering the specific aggregation rules used in the scenario, allowing us to provide either a lower or an upper bound depending on whether the agent forms part of what we call a dominated or dominating path. Being on such a path can be understood as following the beliefs of one's neighbors on it, or leading others. With this, one can even relax the assumptions needed for the convergence rates, such as removing the need for independent signals across agents, requiring only the conditions needed for convergence of beliefs to the true state itself. The speed at which an individual learns is therefore determined not by her own information alone, but by the position she occupies in the chains of influence of her society. Being linked through followership to a well-informed individual accelerates her learning, whereas being bound to poorly informed ones slows it down.

In order to demonstrate the importance of this generalization, we show that our model includes many other rules studied in the literature, and we also provide the conditions under which convergence to the true state is satisfied under many aggregation functions. Here we present the specific case of the constant elasticity of substitution (CES) family. Specifically, for the CES family, we provide the parametrization under the general model and convergence rate bounds depending on the function used, recovering the linear and log-linear rules. Moreover, we find that the best lower bound on the rate in that family of functions corresponds to the harmonic rule.

Our main contributions can be summarized as follows.

\begin{itemize}
    \item We extend non-Bayesian social learning to heterogeneous and time-varying aggregation rules in time-varying adjacency weight matrices, on a fixed graph. We give conditions under which all agents' beliefs converge almost surely to the true state (Theorem~\ref{thm:main}).

    \item We show that the belief dynamics for any differentiable admissible aggregation rule under those conditions can be written exactly as a log-linear update rule (and ultimately as DeGroot) (Theorem~\ref{thm:linearrep}).

    \item We derive an upper bound on the exponential rate of convergence that depends only on how informative the agents are (Theorem~\ref{thm:upperbound}). We then provide more accurate lower and upper bounds for each agent in terms of the most and least informed agent one can reach on specific network paths.

    \item We apply our results in the specific case of the CES family, which contains the linear (DeGroot), logarithmic opinion pool and harmonic rules, among others. We identify the harmonic rule as having the best lower convergence rate bound in the family and show that a single uninformed agent can make the whole network learn at a sub-exponential rate, so that the least informed member of a community may set the pace of learning for everyone.
\end{itemize}

Our paper is related to several areas of the literature. As mentioned before, research on how connected agents can combine dispersed information dates back to the linear opinion-pooling model of DeGroot \cite{degroot1974reaching}. There, agents repeatedly replaces her opinion by a weighted average of those of her neighbors. \cite{golub2010naive} used this rule to characterize when a naive updating leads the group to consensus on the truth state, deriving the wisdom of crowds from the network's influence structure. Some non-Bayesian social learning models extend that idea by combining the aggregation of neighbors' beliefs with the incorporation of new private signals using the Bayesian update. This two-step update was introduced by \cite{jadbabaie2012non}, and later placed on axiomatic foundations by \cite{molavi2017foundations}, further extended in \cite{molavi2018theory} by characterizing the log-linear rule and analyzing the model in time-varying networks. Moreover, specific consensus analysis with expanding information observed by agents were studied by \cite{molavi2013reaching}. Broader surveys of these dynamics are presented in \cite{acemoglu2011opinion}, where they also compare with the Bayesian models, and some advances have avoided the complexity that arises there \cite{ortoleva2024alternatives}. Our work builds directly on \cite{molavi2018theory}, generalizing its convergence and rate results to heterogeneous, time-varying aggregation rules.

Turning to fully Bayesian models, this literature studies agents who update their beliefs taking into account the entire history of observed actions or beliefs in the network, and characterizes when such learning asymptotically aggregates dispersed information \cite{gale2003bayesian, acemoglu2011bayesian}. As noted before, the difficulty with these models is tractability. Some work has been done on analyzing when such models become computationally intractable \cite{rahimian2017complexity}, where the model also corrects the double counting of signals that arrive through several paths. Non-Bayesian models avoid the complexity limitations by discarding the history and learning without recall, formalized in \cite{rahimian2014non}, where the agents take the last-period information shared by their neighbors as sufficient statistics, although this carries the limitation of overweighting some signals.

Apart from analyzing convergence of the beliefs to the true state, it is also important to study the speed at which beliefs concentrate, and to design models that improve it. A substantial literature on distributed estimations and hypothesis testing has quantified this rate. For example, \cite{shahrampour2013exponentially} obtain exponentially fast estimation through distributed dual averaging, while \cite{nedic2017fast, nedic2015nonasymptotic} establish fast and non-asymptotic convergence rates, the latter being over time-varying directed graphs. Another work, \cite{mitra2020new}, proposes a min-based aggregation rule that improves the learning rate, also incorporating Byzantine resilience into the model. In our work, the rate obtained for some instances of the harmonic aggregation function of the CES family in Section \ref{subsec:CES} coincides with the one achieved by their rule. Similar to our general bound, these works tend to express the learning rate through the relative entropies that measure how well agents in the network distinguish the states.

Another area of study is the way information is distributed across agents, and how much agents listen to one another, which also shapes the speed of learning. For instance, \cite{jadbabaie2013information} study how information heterogeneity observed by each agent governs the learning rate, while \cite{golub2012homophily} show that homophily, the tendency of agents to interact more with others who hold similar beliefs, can slow down the learning and the convergence of agents' behaviors to a consensus. Our results complement this view by showing that, beyond who listens to whom, the manner in which opinions are combined determines whether the best- or the least-informed members govern the speed of collective learning.

Finally, as we have seen in the different areas, a growing area of study actually focuses on how beliefs should be aggregated. Some works compare different methods, such as arithmetic (linear) and geometric (log-linear) fusion of beliefs \cite{kayaalp2023arithmetic}, showing that both methods learn the truth exponentially fast, with a faster rate under geometric pooling. Indeed, we recover this comparison in the analysis of the CES family. Other work studies optimal aggregation strategies over graphs \cite{hu2023optimal}. Related to the idea of varying influence and topologies, \cite{kayaalp2026causal} study causal influences among agents and \cite{salami2017social} extend social learning to weakly connected topologies. Following a similar idea, \cite{bordignon2021adaptive} presents the adaptive social learning framework, and its doubly adaptive extension in \cite{carpentiero2025doubly}. In these works, the focus is on evolving environments, enabling agents to adapt to drifting states in non-stationary environments. In this sense, our contribution complements these by allowing each agent to change her entire aggregation function over time, proving convergence and relating the resulting convergence rates to the rules and the informativeness of the signals.

The work is organized as follows. In Section \ref{sec:model} we present the mathematical setup of our general model of NBSL, together with the definitions of the main objects that drive the results and the assumptions required. In Section \ref{sec:thm_univ_bound} we present the main results, the convergence to the true state result, the general expression as a log-linear update rule, and the general convergence rate upper bound. Then, Section \ref{sec:dominating} presents what we call the dominated and dominating paths, which provide tighter convergence rate bounds. Section \ref{sec:applications_and_examples} covers the CES family aggregation rules as a specific case and derive the convergence rates (Section \ref{subsec:CES}), and shows an illustrative example of what the theory can cover (Section \ref{subsec:committee_example}). We then discuss limitations and future work, and conclude, in Sections \ref{sec:discus_and_future} and \ref{sec:conclusions}. All the proofs are given in Appendix \ref{app:one} and Appendix \ref{app:two}.

\section{Model}\label{sec:model}

In this Section we present the mathematical model and the assumptions it needs to satisfy in order to ensure correct learning. 

\subsection{Mathematical Setup}

Consider a fixed directed graph $\cG=(\cN,\cE)$ of $n$ agents, with node set $\mathcal{N}=\{1,2,\dots,n\}$ and edge set $\cE$, learning a finite set of hypotheses $\Theta=\{\theta_1,\theta_2,\dots,\theta_K\}$, also referred to as states. From this set, there exists a unique true state of the world $\theta^*$. This learning is represented in each agent by updating her belief in each hypothesis. In detail, we denote by $\mu_{it}(\theta)\in[0,1]$ the belief of agent $i$ on the state $\theta\in\Theta$ at time $t\ge0$, and we consider that the agent has learned the state $\theta$ when $\mu_{it}(\theta)=1$. This is a probability distribution over $\Theta$ ($\mu_{it}\in\Delta\Theta$), that is, $\sum_{\theta\in\Theta}\mu_{i,t}(\theta)=1$. We denote by $\mu_t(\theta)$ the vector $(\mu_{1,t}(\theta),\dots,\mu_{n,t}(\theta))$.

At each discrete time $t\ge0$, each agent $i\in\mathcal{N}$ observes a private signal $\omega_{i,t+1}\in\Omega_i$ from a finite set of signals $\Omega_i$. The observed signal is drawn from a private probability distribution $l_i^{\theta^*}(\omega_{i,t+1})$, conditioned to the true state of the world\footnote{For all $\theta\in\Theta$, we denote by $\mathbb{P}^{\theta}$ the measure generated by $\theta$.}. Agent $i$ knows the distribution $l_i^{\theta^*}(\omega_{i,t+1})$, as well as the distribution of signals she would observe if any other state $\hat{\theta}\in\Theta$ were the true hypothesis, i.e. $l_i^{\hat{\theta}}\in\Delta \Omega_i$. Thus, the agent can learn and identify a state through the signals observed. However, an agent $i$ may not be able to distinguish the true state from another state if the two states induce the same likelihood distribution to her, i.e. there exists $\hat{\theta}\neq\theta^*$ with $l_i^{\hat{\theta}}(\omega)=l_i^{\theta^*}(\omega)$ for all $\omega\in\Omega_i$. 

To circumvent this, agents can also communicate their beliefs on each hypothesis $\mu_{it}(\theta)$ with one another. In this work, we are interested in the process of aggregating the beliefs of the agents neighbors (including her own) $\mathcal{N}_i = \{i\}\cup\{j\in\mathcal{N}:(i,j)\in\mathcal{E}\}$. The process is described with a possibly time-varying aggregation function for each agent. 

\begin{definition}[Aggregation functions]\label{def:agg}
We define the aggregation function of agent $i$ at each time $t\ge0$ as $\psi_{i,t}:(0,\infty)^n\longrightarrow(0,\infty)$\footnote{This definition is generic for any variables in $(0,\infty)$. However, for the case of aggregating beliefs, the function is $\psi_{i,t}(\mu_t):(0,1]^n\longrightarrow(0,1]$}. The function is twice differentiable and the derivatives of order 2 are continuous (i.e. is in the $C^2$ class). Moreover, we do not consider stochastic functions, i.e. the assignment of functions to agents and times is fixed in advance. 
\end{definition}

Throughout, $x=(x_1,\dots,x_n)\in(0,\infty)^n$ denotes a generic argument of $\psi_{i,t}$. In the belief update, $\psi_{i,t}$ is evaluated, for each $\theta\in\Theta$, at $\mu_t(\theta):=(\mu_{1,t}(\theta),\dots,\mu_{n,t}(\theta))$, which lies in $(0,1]^n$.

Considering everything, an agent $i$ updates her beliefs at time $t$ in two steps. First, she aggregates her neighbors' beliefs with the aggregation function $\psi_{i,t}$, creating an interim belief for each hypothesis $\zeta_{it}(\theta)$:

\begin{equation}\label{eq:interim}
    \zeta_{it}(\theta):=\frac{\psi_{i,t}\bigl(\mu_t(\theta)\bigr)}{S_{it}},
    \qquad
    S_{it}:=\sum_{\tilde\theta\in\Theta}\psi_{i,t}\bigl(\mu_t(\tilde\theta)\bigr).
\end{equation}

Second, she observes the signal $w_{i,t+1}$ and updates her belief by Bayes' rule:

\begin{equation}\label{eq:bayes}
    \mu_{i,t+1}(\theta):=\frac{\zeta_{it}(\theta) \ \ell^{\theta}_i(\omega_{i,t+1})}{m_{it}(\omega_{i,t+1})},
    \qquad
    m_{it}(\omega):=\sum_{\tilde\theta\in\Theta}\zeta_{it}(\tilde\theta)\,\ell^{\tilde\theta}_i(\omega).
\end{equation}

Notice that this update is not fully Bayesian. As in \cite{molavi2018theory}, agents respect imperfect recall, meaning that they treat their neighbors' last-period beliefs as sufficient statistics for the entire history. Apart from this, agents only communicate their current beliefs, while signals, likelihoods, and their aggregation functions are kept private.

\begin{remark}
    The standard alternative to NBSL is fully Bayesian learning, in which agents apply the Bayes rule to every signal they observed and on all the history of her neighbors' beliefs since the beginning of the process, thus not respecting the imperfect recall assumption. This method is often intractable in practice and implies redundancy (double counting) of information, because the same signal observed by agent $i$ can reach an agent $j$ through multiple paths in the information sharing process. 
\end{remark}

\subsection{Definitions and Assumptions}

In this Section, we present the main definitions and objects used throughout the article and the assumptions related to them under which the model makes every agent learn the true state almost surely, that is, $\lim_{t\to\infty}\mu_{i,t}(\tstar)=1$ for all $i\in\mathcal{N}$. There are four assumptions, one on the aggregation functions, one on the communication graph, one on the likelihoods and one on the initial beliefs.

\subsubsection{Aggregation Functions and Network Weights}

As mentioned, the focus of this article is on the generalization of the aggregation function to create the interim belief, allowing for heterogeneity and time-varying (non-stochastic) functions. Not every function is valid, however, if we want the agents to learn the true state. To this end, we restrict the functions to a family of them, determined by what we define as logarithmic curvature. Before defining it, we first present other definitions and formalize the trust weights each agent assigns to others' beliefs.

To this end, we introduce the elasticity, a measure of the influence of agent $j$ on agent $i$.
 
\begin{definition}[Influence]\label{def:elasticity}
For $x\in(0,\infty)^n$, the influence of agent $j$ on agent $i$ at time $t$ is the elasticity

\begin{equation*}
    e_{ij,t}(x):=\frac{\partial\log\psi_{i,t}}{\partial\log x_j}(x)
    =\frac{x_j\,\partial_{x_j}\psi_{i,t}(x)}{\psi_{i,t}(x)} .
\end{equation*}
\end{definition}

It is the relative change of the aggregate caused by a relative change in $x_j$: if agent $j$ becomes $1\%$ more confident in state $\theta$ (i.e. $\mu_{j,t}(\theta)$ increases by $1\%$), agent $i$'s aggregate confidence in $\theta$ ($\psi_{i,t}(\mu_t(\theta))$) rises by approximately $e_{ij,t}(x)\%$. It is therefore the effective weight that agent $i$ gives to agent $j$'s belief in $\theta$ at the belief profile $x$. It thus measures social influence in its usual sense: the extent to which a change in one person's opinion moves another's.

Then, evaluating it at the profile where all beliefs coincide $\mathbf{1}=(1,\dots,1)$ yields a baseline measure of trust, which we formally define as the weights of the network.
 
\begin{definition}[Graph weights]\label{def:weights}
The trust weight that agent $i$ places on agent $j$ at time $t$ is

\begin{equation}\label{eq:tangentweights}
    a_{ij,t+1}:=\frac{\partial\psi_{i,t}}{\partial x_j}(\mathbf{1}),
\end{equation}

which coincides with $e_{ij,t}(\mathbf{1})$ when $\psi_{i,t}(\mathbf{1})=1$.
\end{definition}

These weights play the role of the adjacency matrix in most models, except that here they are derived from the aggregation function rather than imposed. For them to be interpretable as weights on the network $\cG$, an agent should place positive weight only on her neighbors.

Trust weights describe the aggregation when beliefs are in consensus. If we move away from consensus, an agent may reallocate influence among her neighbors depending on what they believe. This is captured by the cross-derivatives:

\begin{equation*}
    H^{i,t}_{jk}(x):=\frac{\partial^2\log\psi_{i,t}}{\partial\log x_j\,\partial\log x_k}(x)
    =\frac{\partial e_{ij,t}}{\partial\log x_k}(x),
\end{equation*}

which measure the change in $j$'s influence on $i$ when $k$ revises her belief in $\theta$. We normalize them by the two influences involved and obtain the first main object of the article.
 
\begin{definition}[Logarithmic curvature]\label{def:curv}
For distinct $j,k\in\mathcal{N}_i$ with $e_{ij,t}(x)>0$ and $e_{ik,t}(x)>0$, the \emph{logarithmic curvature} of $\psi_{i,t}$ is:

\begin{equation}\label{eq:curvature}
    \kappa_{ijk,t}(x):=-\frac{H^{i,t}_{jk}(x)}{e_{ij,t}(x)\,e_{ik,t}(x)} .
\end{equation}
\end{definition}

The curvature describes how the agent's trust weights vary when neighbors change their beliefs.\footnote{It also measures how far $\psi_{i,t}$ is from log-linear in $(x_j,x_k)$.} If $\kappa>0$, an increase in agent $k$'s belief in $\theta$ lowers the weight on agent $j$'s belief, so the aggregate tilts toward the neighbors that are more confident in $\theta$. If $\kappa<0$, it tilts toward the less confident ones. If $\kappa=0$, the weight on $j$ does not depend on $k$'s belief. For instance, for a weighted power mean $\big(\sum_j a_j x_j^{\rho}\big)^{1/\rho}$ one has $\kappa\equiv\rho$. Thus the arithmetic mean has $\kappa=1$, the geometric mean has $\kappa=0$, and the harmonic mean has $\kappa=-1$. Thus, positive curvature describes individuals whose confidence in a hypothesis is pulled toward the neighbors most convinced of it. Negative curvature describes individuals who give more weight to the neighbors who most doubt a hypothesis. The log-linear rule, for instance, is the neutral case, in which influence does not depend on the configuration of opinions.

Once we have described the main object and presented the weights in the graph, we now present the assumptions made on the aggregation functions and, equivalently, on the trust weights.

\begin{hypothesis}\label{hyp:H1}
For all $t\ge 0$ and every agent $i\in\mathcal{N}$, the function $\psi_{i,t}$ satisfies:
\begin{enumerate}[label=(\roman*)]
    \item \emph{Monotonicity (only neighbors matter):} $\psi_{i,t}$ is strictly increasing in the argument of each neighbor $j\in\mathcal{N}_i$, i.e.\ $\partial_{x_j}\psi_{i,t}>0$ on $(0,\infty)^n$, and does not depend on agents outside agent $i$'s neighborhood, i.e.\ $\partial_{x_j}\psi_{i,t}\equiv 0$ for $j\notin\mathcal{N}_i$;
    \item \emph{Homogeneity of degree one (no scale effects):} scaling all arguments by the same factor scales the output by that factor, i.e.\ $\psi_{i,t}(\lambda x)=\lambda\,\psi_{i,t}(x)$ for all $x\in(0,\infty)^n$ and $\lambda>0$;
    \item \emph{Normalization:} $\psi_{i,t}(1,\dots,1)=1$;\footnote{In view of \eqref{eq:interim}, the value of $\psi_{i,t}(1,\dots,1)$ does not affect the results, we fix it for convenience.}
    \item \emph{Bounded curvature (limited interaction between neighbors):} for all $j\neq k$ and all $x\in(0,\infty)^n$,    
    \begin{equation}\label{eq:constraint}
        |H^{i,t}_{jk}(x)|\le e_{ij,t}(x)\,e_{ik,t}(x),
    \end{equation}    
    so the cross-effect between any two agents' inputs on agent $i$ cannot exceed the product of the individual influences those agents exert on $i$.
\end{enumerate}
\end{hypothesis}

\begin{remark}
    Condition (i) says that an agent becomes more confident in a state when any of her neighbors does, and ignores non-neighbors. Hence $a_{ij,t+1}>0$ if and only if $j\in\mathcal{N}_i$, i.e. an agent only assigns a positive trust weight to her neighbors on the network $\cG$ and to her own beliefs $a_{ii,t+1}>0$. Condition (ii) says that only relative confidence matters, so that the trust weights sum to one $\sum_j a_{ij,t+1}=1$, $\forall i\in \mathcal{N}_i$. Condition (iv) is equivalent to $|\kappa_{ijk,t}|\le1$ whenever it is defined\footnote{If it is not defined, it means that $j$ or $k$ are not both neighbors of $i$ and therefore the inequality reads as $0\le0$ which is trivially satisfied. In contrast, it holds trivially when $j$ or $k$ is not a neighbor of $i$.}, i.e., whenever both $j$ and $k$ are neighbors of $i$. Without this condition, an agent could follow only the most extreme opinions in her neighborhood, and a single misguided neighbor could dominate her beliefs.\footnote{Condition (iv) directly implies the concavity of $\psi_{i,t}$ and its dual (defined in Remark \ref{rem:dual}), which are important facts, because, combined with homogeneity of degree one and Euler's formula, it implies that $\psi_{i,t}$ and its dual are below their tangent hyperplanes at $\mathbf{1}=(1,\dots,1)\in(0,\infty)^n$ creating some form of linearization, more details are presented in Appendix \ref{app:one}, which is dedicated to the proof of Theorem $\ref{thm:main}$.}
\end{remark}

The second assumption, similar to inequalities (8) and (9) of \citet{molavi2018theory}, ensures that information can travel through the network and that agents keep listening to each other.
 
\begin{hypothesis}\label{hyp:H2}
The network $\cG$ is strongly connected, and there are constants $0<\aund\le\abar<1$ and a sequence $\lambda_t\in(0,1)$ with $t\lambda_t\to\infty$ such that, for all $t\ge0$ and all $i\in\cN$,

\begin{equation*}
    a_{ij,t+1}\ge \aund\,\lambda_t \ \ \text{for } j\in\mathcal{N}_i\setminus\{i\},
    \qquad
    \sum_{k\neq i}a_{ik,t+1}\le \abar\,\lambda_t .
\end{equation*}
\end{hypothesis}

Strong connectivity means that every agent is influenced, directly or indirectly, by every other, so that private information observed anywhere can reach everyone through beliefs. The sequence $\lambda_t$ determines how much an agent listens to others relative to herself\footnote{All trust placed on neighbors is of order $\lambda_t$, while self-trust satisfies $a_{ii,t+1}\ge1-\abar\lambda_t>0$}. This disposition to listen to her neighbors may decrease over time, as agents become more self-reliant as time passes. However, the rate of decrease is limited ($t\lambda_t\to\infty$), since an agent who stops listening too early may lock in before the information in the network has reached her. This effect is not private: an individual who becomes self-reliant too early harms not only her own learning but also that of everyone who listens to her, since she ceases to relay the information held elsewhere in the network. Apart from this, a constant $\lambda_t=\lambda$ is also admissible, as in models with time-invariant adjacency matrices.
 
Together, Hypotheses \ref{hyp:H1} and \ref{hyp:H2} describe agents who learn from each other through a fixed network, by exchanging only their current beliefs. Since aggregation functions may change over time, agents may revise how much they trust their neighbors and themselves, which captures phenomena such as confirmation bias or a growing distrust of the network. However, the network remains fixed, what changes is how much each agent listens to whom.

\subsubsection{Private Information and Entropy}

We now focus on the signals and the informativeness of each agent according to her likelihoods of observing each signal under the different states. First, we provide the assumptions made on the information received by each agent.

\begin{hypothesis}\label{hyp:H3}
The likelihoods have full support, $\ell^{\theta}_i(\omega)>0$ for all $i,\theta,\omega$. Moreover, for every $\theta\neq\theta^*$ there is an agent $i$ with $\ell^{\theta}_i\neq\ell^{\theta^*}_i$. Let $\Theta_i:=\{\theta\in\Theta:\ell^{\theta}_i\neq\ell^{\theta^*}_i\}$ denote the set of states that agent $i$ can distinguish locally from $\theta^*$. Finally, the observations $\omega_{i,t}$ are independent across agents at each fixed $t$ and i.i.d. across time.
\end{hypothesis}

These are standard Hypotheses concerning the likelihoods used by the agents to infer the true state. In practice, since no agent knows which state is the true state, the second condition should read: for all paris of different states $\theta_1$ and $\theta_2\in\Theta$, there is an agent $i\in\mathcal{N}$ such that $\ell^{\theta_1}_i\neq\ell^{\theta_2}_i$. This condition means that, in order to infer the true state, it is necessary, for all the other possible states, that at least one agent must be able to distinguish this state from the real underlying state. Otherwise, it would be impossible to learn. This is the formal counterpart of a division of cognitive labor, in which no individual needs to be able to answer every question, provided that, for each question, someone in the society can. The third Hypothesis states the independence of the received signals across agents, which is a strong assumption that is used only once in the proof of Theorem \ref{thm:main}.

Once we have presented these conditions, we can now define a measure of how much information an agent received after observing a signal. This relates to her informativeness and is given by the relative entropy.

\begin{definition}[Relative entropy]\label{def:entropy}
For $i\in\mathcal{N}$ and $\theta,\that\in\Theta$, we define the relative entropy by

\begin{equation*}
    h_i(\theta,\hat\theta)=\sum_{\omega\in\Omega_i}\ell^{\theta}_i(\omega)\log\frac{\ell^{\theta}_i(\omega)}{\ell^{\hat\theta}_i(\omega)}=\mathbb{E}^{\theta}\Bigl[\log\frac{\ell^{\theta}_i(\omega_{i,t})}{\ell^{\hat\theta}_i(\omega_{i,t})}\Bigr]
\end{equation*}

for all $t>0$ (since the signals are assumed independent and identically distributed across time, we may simply take $t=1$). We write $h_i(\that):=h_i(\tstar,\that)$ for the relative entropy with respect to the true state.
\end{definition}

This term represents the learning capacity of an agent $i$, how efficiently she is capable of distinguishing one state from another. In practice, an agent can be very informative while others are not. This affects the local update rule using the Bayes formula, but learning also occurs through communication and aggregation. However, since relative entropies are finite ($l_i^\theta(w)>0, \forall i,w,\theta$) and the number of agents from which she can aggregate information is also finite, an infinitely informed agent does not exist, and hence there is an upper bound on the speed at which they can learn.

Then, we can also define the following.

\begin{definition}[Maximum informativeness]\label{def:max_inform}
    We define the maximum informativeness among the agents and the informativeness of a hypothetical agent receiving, at each period, the most informative signal among all the agents in the network, respectively, as follows:
    \begin{equation}
        h_\star(\that):=\max_{j\in\cN}h_j(\that)\ ,
        \qquad 
        \bar h(\that):=\E^{\tstar}\Bigl[\max_{j\in\mathcal{N}}\log\frac{\ell^{\tstar}_j(\omega_{j1})}{\ell^{\that}_j(\omega_{j1})}\Bigr].
    \end{equation}
\end{definition}

\subsubsection{Initial Beliefs}

Finally, we assume that the initial beliefs are positive values. In doing so, we ensure that no division by zero occurs when using \eqref{eq:interim} and \eqref{eq:bayes} and no state is rejected from the start.\footnote{Note that, by induction, applying the learning algorithm, beliefs will remain strictly positive for all $t\ge0$, avoiding any division by 0.}

\begin{hypothesis}\label{hyp:H4}
The priors are deterministic with $\mu_{i,0}(\theta)>0$ for all $i,\theta$.
\end{hypothesis}

\section{Convergence, Log-Linearity and General Convergence Rate Bound}\label{sec:thm_univ_bound}

In this Section we present our main result which states the convergence to the true state $\tstar$ for every agent, as well as the possibility of expressing our model in the form of the log-linear learning rule and a bound on the convergence rate of the beliefs. Throughout the remainder of the Section we assume Hypotheses~\ref{hyp:H1} to~\ref{hyp:H4} for clarity, we could obtain, however, the same results for the log-linear expression (Theorem \ref{thm:linearrep}) and the convergence rate bound (Theorem \ref{thm:upperbound}) with less restrictive assumptions.

The main Theorem of the article reads as follows.

\begin{theorem}\label{thm:main}
Under Hypotheses \textup{\ref{hyp:H1}} to \textup{\ref{hyp:H4}}, every agent $i \in \cN$ will learn almost surely the true state:
\begin{equation*}
    \lim_{t\to\infty}\mu_{it}(\theta^*)=1
    \qquad
    \Pstar\text{-almost surely, for every } i\in\cN .
\end{equation*}
\end{theorem}

\begin{proof}
See Appendix \ref{app:one}.
\end{proof}

Theorem~\ref{thm:main} states that heterogeneity in the way individuals combine opinions, and changes in them over time, do not by themselves prevent a society from reaching the truth. As long as no one stops listening to others too quickly and no one can be captured by extreme opinions, dispersed private information is eventually aggregated.

\begin{remark}\label{rem:dual}
    A direct corollary is that the dual function of an admissible sequence $\psi_{i,t}$ (i.e., satisfying Hypotheses~\ref{hyp:H1} and~\ref{hyp:H2}) $\phi_{i,t}(x_1,\dots,x_n):=\bigl(\psi_{i,t}(1/x_1,\dots,1/x_n)\bigr)^{-1}$ is also an aggregation function satisfying Hypotheses~\ref{hyp:H1} and~\ref{hyp:H2}. Therefore, the sequence $\phi_{i,t}$ for $i\in\mathcal{N}$ and $t\ge0$ is admissible for Theorem~\ref{thm:main}, where the trust weights are the same and the logarithmic curvature is simply the opposite. We shall see this duality, for instance, in the case of the CES family in Section \ref{subsec:CES}. In social terms, for every norm that leans toward the most confident opinions there is a mirror norm leaning toward the most doubtful ones, and both are equally compatible with collective learning.
\end{remark}

Before presenting the convergence rate bounds, it is interesting and useful for the proof to see that the model we presented, which accepts heterogeneous and time-varying (non-stochastic) aggregation functions, can be expressed as a log-linear learning rule with stochastic adjacency weights, which we call logarithmic weights.

\begin{definition}[Logarithmic weights]\label{def:logweights}
Fix $\that\neq\tstar$ and $t\ge0$. Defining the log-segment

\begin{equation*}
    u(s):=(1-s)\log\mu_t(\that)+s\log\mu_t(\tstar)\in\R^n,\qquad s\in[0,1],
\end{equation*}

we define the logarithmic weights, which are the entries of a stochastic matrix $\tilde{E}_t(\that)$, as:

\begin{equation*}
    \tilde e_{ij,t}:=\int_0^1 e_{ij,t}\bigl(e^{u(s)}\bigr)\,ds .
\end{equation*}
\end{definition}

Then, we define the log ratios of the beliefs and the likelihoods with respect to the true state as follows:

\begin{equation*}
    x_{it}(\that):=\log\frac{\mu_{it}(\tstar)}{\mu_{it}(\that)}\ ,
    \qquad
    y_{i,t+1}(\that):=\log\frac{\ell^{\tstar}_i(\omega_{i,t+1})}{\ell^{\that}_i(\omega_{i,t+1})}.
\end{equation*}

We can then express the model update as follows.

\begin{theorem}\label{thm:linearrep}
$\mathbb{P}^{\tstar}$-almost surely, for all states $\that\ne\tstar$:
\begin{itemize}
\item[(i)] $x_{t+1}(\that)=\tilde{E}_t(\that)x_t(\that)+y_{t+1}(\that)$, exactly.
\item[(ii)] $\tilde{E}_t(\that)$ is non-negative and row-stochastic, and is
$\mathcal F_t$-measurable.
\end{itemize}
\end{theorem}

\begin{proof}
See Appendix \ref{app:T2}.
\end{proof}

One can observe that $(i)$ has the same structure as the log-linear learning rule, and we exactly recover it by taking the weights as deterministic instead of stochastic, and equivalently, $\tilde{E}_t$ is equal to the adjacency matrix of the graph when every agent uses the log-linear rule.

The main utility of this result is to prove Theorem \ref{thm:upperbound}. This result is worth mentioning because, on the one hand, the structure of part $(i)$ provides an elegant way to represent how information is distributed in the communication graph, and on the other hand, we believe that the matrix $\tilde{E}_t(\that)$ is an interesting object that has not appeared in the literature.

With this in mind, it is easy to prove the general upper bound on the exponential speed of convergence.

\begin{theorem}\label{thm:upperbound}
For all states different from the true state ($\that\ne\tstar$), every agent's rate of convergence of beliefs $\mu_{it}(\that)$ is almost surely bounded above by the informativeness of the hypothetical agent of Definition \ref{def:max_inform} $\bar h(\that)$:

\begin{equation*}
    \limsup_{t\to\infty}\frac{x_{it}(\that)}{t}\;\le\;\bar h(\that)
    \qquad\text{and}\qquad
    \limsup_{t\to\infty}\frac1t\bigl|\log\mu_{it}(\that)\bigr|\;\le\;\bar h(\that)
    \qquad\mathbb{P}^{\tstar}\text{-a.s.}
\end{equation*}
\end{theorem}

\begin{proof}
See Appendix \ref{app:T3}.
\end{proof}

In other words, the rate at which every agent's beliefs in the incorrect states decrease, i.e. the learning rate of the true state, is capped by the informativeness of a hypothetical agent using the most informative signal at each time.\footnote{This is consistent with Theorem 8 of \citet{molavi2018theory}, because the rate they obtain is indeed smaller than the maximum of the relative entropies.} It also draws a direct parallel between the ability to distinguish states and the convergence rate: the better an agent can distinguish the true state from a false one (that is, the higher the relative entropy), the faster her belief in that false state decreases. This provides a benchmark against which any form of social organization can be assessed. Explained differently, the best a group can achieve is to behave as if each member could, at every period, consult the most informative piece of evidence available anywhere in the society. Better communication can bring a group closer to this frontier, but only better information can move the frontier itself.

\section{Convergence rate on dominated and dominating paths}\label{sec:dominating}
In this Section, we present tighter convergence rates when one considers the aggregation functions used by each agent. For clarity, we first introduce the notions of dominated and dominating paths, present the convergence rate induced by the existence of such paths, and then show how to infer the paths for each agent according to the logarithmic curvature of the aggregation function they use. These results will also enable us to draw direct corollaries for the well-known CES family in Section \ref{subsec:CES}, and more generally to a subclass of functions with specific logarithmic curvatures. 

We fix a state $\hat{\theta}\in\Theta$ different from the true state $\tstar$ throughout this Section. We again assume that the agents satisfy Hypotheses~\ref{hyp:H1}--\ref{hyp:H4}.\footnote{We could instead suppose they satisfy a strictly weaker set of constraints (implied by Hypotheses~\ref{hyp:H1}--\ref{hyp:H4}), but we prefer to preserve clarity, more details are presented just before the proofs (see Hypothesis \ref{hyp:R} and the remark that follows).}

To discuss about dominating and dominated paths, we first define dominated and dominating neighbors, with respect to a constant $c>0$.

\begin{definition}[Domination]\label{def:domination}
Let $c>0$, $i\in\cN$, $t\ge0$ and let $j\in\cN$ be a neighbor of $i$. We say that

\begin{align}
  \text{$\psi_{i,t}$ is \emph{$c$-dominating in $j$}}\quad&\text{if}\quad
  \psi_{i,t}(x)\;\le\;a_{ij,t+1}^{-1/c}\;x_j
  \qquad\text{for every }x\in(0,\infty)^n, \label{eq:dominating}\\[2pt]
  \text{$\psi_{i,t}$ is \emph{$c$-dominated in $j$}}\quad&\text{if}\quad
  \psi_{i,t}(x)\;\ge\;a_{ij,t+1}^{\;1/c}\;x_j
  \qquad\text{for every }x\in(0,\infty)^n. \label{eq:dominated}
\end{align}
\end{definition}

When beliefs are involved and the objective is to lower $\psi_{i,t}$ (we have fixed a wrong state $\hat\theta\neq\theta^\star$), \eqref{eq:dominating} can be read as an alignment of agent $i$ with the belief of her neighbor $j$. In the same way, \eqref{eq:dominated} means that the belief of agent $i$ is followed by her neighbor $j$. Moreover, note that, because of homogeneity of degree one, these constraints read the same if we replace $x\in(0,\infty)^n$ by $x\in(0,1]^n$. 

The strongly connected graph links all the agents with each other, and therefore, the definitions can be extended to paths as follows.

\begin{definition}[Chains and reachable sets]\label{def:reachable}
Let $c>0$. An agent $j$ is \emph{$c$-self-dominating} if $\psi_{j,t}$ is $c$-dominating in $j$ for every $t\ge0$\footnote{A posteriori, after $T\ge0$ steps in the NBSL algorithm, Hypothesis \ref{hyp:H4} still holds (deterministic, a posteriori, after observing the $T$ signals). So, the belief sequence starting at $T$ satisfies all the hypotheses of our framework. Therefore, our results concerning convergence and the speed of this convergence still hold, if we replace $t\ge0$ by $t\ge T$ in definition \ref{def:reachable}.}, and a directed edge $(q,q')\in\cE$ is \emph{$c$-dominating} if $\psi_{q,t}$ is $c$-dominating in $q'$ for every $t\ge0$. A directed path is $c$-dominating if all of its edges are. Set

\begin{equation*}
    R_c(i):=\bigl\{\,j\in\cN:\ j\ \text{is $c$-self-dominating, and some $c$-dominating directed path joins $i$ to $j$ in }\cG\,\bigr\},
\end{equation*}

Note that $i\in R_c(i)$ whenever $i$ is $c$-self-dominating. The set $S_c(i)$ is defined identically, with ``dominated'' in place of ``dominating'' in the definition.
\end{definition}

We can then extend our interpretation to the paths. If there exists a dominating path from $i$ to $j$, it means that the belief of $i$ will follow the belief of $j$, and conversely for dominated paths. 

Our next results depend on the following quantity $\beta_{\lambda}:=\limsup_{t\to\infty}\frac1t\sum_{\tau\le t}\bigl|\log(1-\bar a\lambda_\tau)\bigr|$. This can be understood as the long-run attention paid to neighbors. In detail, when $\lambda_t$ is low, the weights assigned to the neighbors (and $\beta_{\lambda}$) are low as well. Conversely, when $\lambda_t$ is high, the long run weights on neighbors and the $\beta_{\lambda}$ are also high.

This machinery enables us to derive convergence rates linking the sets $R_c(i)$ and $S_c(i)$, with the relative entropies of agents in the graph, through the following Theorems.

\begin{theorem}[Lower bound]\label{thm:lower}
Suppose Hypotheses~\ref{hyp:H1} to \ref{hyp:H4} hold. Then, for every agent $i\in\cN$, and every $\that\neq\tstar$,

\begin{equation}\label{eq:lower}
  \liminf_{t\to\infty}\frac1t\bigl|\log\mu_{it}(\that)\bigr|
  \;\ge\;\sup_{c>0}\bigl(\max_{j\in R_c(i)}h_j(\that)\;-\;\frac{\beta_{\lambda}}{c}\bigr)
  \qquad\Pstar\text{-almost surely},
\end{equation}

with the convention $\max_{j\in\emptyset}h_j(\that)=-\infty$, so that the bound is void when $R_c(i)=\emptyset$.
\end{theorem}

We prove this Theorem in Appendix \ref{app:T4}. It tells us that each agent aligns with the best agent she can reach with a dominating path. The term $\beta_{\lambda}$ diminishes the rate because the best agent in $R_c(i)$ will be influenced by her neighbors, and this influence is captured by the term $\beta_{\lambda}/c$. An individual can thus learn as fast as the best expert she is linked to through a chain of followership, even if she is not directly connected to that expert and is herself uninformed. The term $\beta_\lambda/c$ is the cost of the expert's own attention to others, which dilutes the value of her private evidence.

\begin{theorem}[Upper bound]\label{thm:upper}
Suppose Hypotheses~\ref{hyp:H1} to \ref{hyp:H4} hold. Then, for every agent $i\in\cN$, and every $\that\neq\tstar$,

\begin{equation}\label{eq:upper}
  \limsup_{t\to\infty}\frac1t\bigl|\log\mu_{it}(\that)\bigr|
  \;\le\;\inf_{c>0}\bigl(\min_{j\in S_c(i)}h_j(\that)\;+\;\frac{\beta_{\lambda}}{c}\bigr)
  \qquad\Pstar\text{-almost surely},
\end{equation}

with the convention $\min_{j\in\emptyset}h_j(\that)=+\infty$, so that the bound is void when $S_c(i)=\emptyset$.
\end{theorem}

We prove this Theorem in the Appendix \ref{app:T5}. The interpretation of this Theorem is similar to the interpretation of Theorem~\ref{thm:lower}, it means that an agent $i$, giving importance to all the agents in $S_c(i)$ ends up with a rate lower than or equal to that of the least informed agent in $S_c(i)$ up to an additional term. This agent also captures information from her neighbors, which is represented by the additional term $\beta_{\lambda}/c$. Symmetrically, the least informed person an individual cannot discount sets a ceiling on how fast she can learn.

\begin{remark}
The reader may wonder whether we could combine these two results to obtain upper and lower bounds on the convergence rate. This is the object of the following Proposition, stating that if there exist $c_1>0$ and $c_2>0$ such that $R_{c_1}(i)\neq \emptyset$ and $S_{c_2}(i)\neq \emptyset$, then $R_{c_1}(i)=S_{c_2}(i)=\{r\}$ with $r$ some agent in the graph. 

\begin{proposition}[Equal sets]\label{prop:S=R} 
    Under hypotheses \ref{hyp:H1} to \ref{hyp:H4}, if there exist $c_1>0$ and $c_2>0$ and an agent $i\in\mathcal{N}$, such that both $R_{c_1}(i)\neq \emptyset$ and $S_{c_2}(i)\neq \emptyset$, then, 
    
    \begin{itemize}
        \item there exists an agent $r\in\mathcal{N}$ which is both $c_1$-self-dominating and $c_2$-self-dominated such that $R_{c_1}(i)=S_{c_2}(i)=\{r\}$.
        \item $\Pstar$-almost surely
            \begin{equation*}
                h_r(\that)-\frac{\beta_\lambda}{c_1}
                \;\le\;\liminf_{t\to\infty}\frac1t\bigl|\log\mu_{it}(\that)\bigr|
                \;\le\;\limsup_{t\to\infty}\frac1t\bigl|\log\mu_{it}(\that)\bigr|
                \;\le\;h_r(\that)+\frac{\beta_\lambda}{c_2},
            \end{equation*}
            
            and the rate is exactly $h_r(\that)$ when $\beta_\lambda=0$.
    \end{itemize}
\end{proposition}

The proof is presented in the Appendix \ref{app:P1}.
\end{remark}

Under which conditions on the aggregation functions do dominating or dominated paths exist? Our goal, to finish this Section, is to link the existence of these paths to the logarithmic curvatures of the aggregation functions, by presenting a sufficient condition to the existence of a path from an agent $i$ to an agent $j$. By doing so, it will make the Theorem more convenient to use in some cases, for example it will give us direct corollaries for the CES family in the next Section.

\begin{proposition}[Curvature]\label{prop:curv}
Suppose Hypothesis~\ref{hyp:H1} holds and let $j$ be a neighbor of $i$.

\begin{itemize}
    \item[(i)] If $H^{i,t}_{jl}(x)\ge c\,e_{ij,t}(x)e_{il,t}(x)$ for all $l\neq j$ and all
    $x>0$, (equivalently $\kappa_{i,j,l,t}\le-c$ when $l$ is also a neighbor of $i$) then \text{$\psi_{i,t}$ is \emph{$c$-dominating in $j$}}.
    \item[(ii)] If $H^{i,t}_{jl}(x)\le-c\,e_{ij,t}(x)e_{il,t}(x)$ for all $l\neq j$ and all
    $x>0$ (equivalently $\kappa_{i,j,l,t}\ge c$ when $l$ is also a neighbor of $i$) then \text{$\psi_{i,t}$ is \emph{$c$-dominated in $j$}}.
\end{itemize}
\end{proposition}

We prove this Proposition in the Appendix \ref{app:P2}. We can then interpret the logarithmic curvature as a relative attention to a neighbor of $i$ with respect to the others. If $(i)$ or $(ii)$ hold for all $t$, in a path from an agent $i$ to an agent $j$, then we can conclude, by definition, that $j$ belongs to $R_c(i)$ or $S_c(i)$ respectively. Moreover, a direct corollary is that taking the dual swaps $R_c(i)$ and $S_c(i)$. Curvature therefore summarizes a behavioral trait, namely the extent to which an individual lets the most skeptical or the most convinced among her contacts drive her opinion. Proposition~\ref{prop:curv} shows that such traits, and not the network structure alone, determine who effectively leads and who effectively lags in collective learning.

\section{Applications and Examples}\label{sec:applications_and_examples}

\subsection{Constant Elasticity of Substitution family}\label{subsec:CES}

This Section focuses on the well-known CES family. We show that it can be treated under the framework of the preceding Section, and therefore derive bounds on the convergence rates for this family. This family includes the log-linear, linear, and harmonic rules. We assume Hypotheses \ref{hyp:H2}, \ref{hyp:H3} and \ref{hyp:H4}, and show that this family satisfies Hypothesis \ref{hyp:H1}, so Theorem \ref{thm:main} applies.

These are the aggregators of the following form, for a parameter $\rho\in[-1,1]$ and weights $b_{i,j,t}$ such that $\sum_{j=1}^nb_{i,j,t}=1$:

\[
    \psi(x_1,\ldots,x_n)=
    \begin{cases}
        \left(\displaystyle\sum_{j=1}^{n}b_{i,j,t} x_j^{\rho}\right)^{1/\rho},
        & \text{if } \rho\neq 0,\\[2ex]
        \displaystyle\prod_{j=1}^{n}x_j^{b_{i,j,t}},
        & \text{if } \rho=0.
    \end{cases}
\]

As explained by \citet{molavi2018theory}, this class of functions satisfies Hypothesis~\ref{hyp:H1} because its logarithmic curvature is always equal to $\rho$. Moreover, for this family we have $a_{i,j,t+1}=b_{i,j,t}$, meaning that the weights involved in Hypothesis~\ref{hyp:H2} are the same as the weights defining the function. We may then fix the weights so that $\psi_{i,t}$ satisfies Hypothesis~\ref{hyp:H2}. Note that for $\rho=1$ we obtain the linear rule and for $\rho=0$ we obtain the log-linear rule. These are, respectively, the objects of Theorems 5 and 4 of \citet{molavi2018theory}, showing that we indeed generalize their results.

We now present two convergence rate results, one for $\rho<0$ and another for $\rho>0$. Our results depend on the quantity $\beta_{\lambda}:=\limsup_{t\to\infty}\frac1t\sum_{\tau\le t}\bigl|\log(1-\bar a\lambda_\tau)\bigr|$.

\begin{corollary}\label{cor:rateneg}
Suppose there exists a fixed $\rho\in[-1,0)$, such that, for all agent $i$ and $t\ge0$, one has $\psi_{i,t}$ belongs to the CES family, with the parameter $\rho_{i,t}$ such that $-1\le\rho_{i,t}\le\rho$. Then, $\mathbb{P}^{\tstar}$-almost surely, for every agent $i$ and every $\that\neq\tstar$,
\begin{equation*}
    \liminf_{t\to\infty}\frac1t\bigl|\log\mu_{it}(\that)\bigr|\;\ge\;h_\star(\that)-\frac{\beta_{\lambda}}{|\rho|} \;
    \qquad\Pstar\text{-almost surely},
\end{equation*}

with $h_\star(\that)=\max_{j\in\mathcal{N}}h_j(\tstar,\that)$.
\end{corollary}

The proof is a direct corollary of Theorem \ref{thm:lower}, indeed, for all $i$ and $j$ in $\mathcal{N}$, there is a path by strong connectivity, since $|\rho|$ is fixed with respect to agent and to time, and different from 0, Proposition \ref{prop:curv} gives us that the path will be $|\rho|$-dominating, then $R_{-\rho}(i)=\mathcal{N}$ for all agent $i$. We note that this bound is better than the one obtained by \citet{molavi2018theory} when $\lambda_t$ converges to $0$. The interpretation is to say that each agent follows the best one, giving a similar aggregation structure than the min rule presented by \cite{mitra2020new}, we indeed obtain a convergence rate which is equivalent. We also observe that, the harmonic rule ($\rho=-1)$ gives the best lower bound for a given fixed decaying sequence $\lambda_t$.

We now present another result, this time for $\rho>0$.

\begin{corollary}\label{cor:ratepos}
Suppose there exist a fixed $\rho\in(0,1]$, such that, for all agent $i$ and $t\ge0$, one has $\psi_{i,t}$ belongs to the CES family, with the parameter $\rho_{i,t}$ such that $\rho\le\rho_{i,t}\le1$. Then, for every agent $i$ and $\hat\theta \neq \theta^\star$,
\begin{equation*}
    \limsup_{t\to\infty}\frac1t\bigl|\log\mu_{it}(\that)\bigr|\;\le\;\underline{h}(\that)+\frac{\beta_{\lambda}}{\rho} \; 
    \qquad\Pstar\text{-almost surely}
\end{equation*}

with $\underline{h}(\that):=\min_{j\in\cN}h_j(\that)$.
\end{corollary}

This is also a direct corollary of Theorem \ref{thm:upper}: this time, one has $S_{\rho}(i)=\mathcal{N}$ for all agents $i$. This time, our result can be interpreted by saying that a blind agent, who cannot distinguish the true state from $\that$, slowly recovers the true state if she does not trust the network sufficiently to rely on the best agents and prefers to balance her update choices on all her neighbors. The result is that all the agents minimize their learning speed, moreover, if there is a blind agent, ($\underline{h}(\that)=0$), and $\beta_\lambda=0$, all the agents in the network learn with a sub-exponential rate of convergence.

We then observe that the convergence rate at which agents recover the true state depends on the quantity of information each agent is capable of providing to herself and the degree of confidence she gives to her neighbors. It is translated in the CES case by the sign of the coefficient $\rho$, and the degree at which $\lambda_t$ decreases. 

\subsection{An illustrative example: a committee assessing a project}
\label{subsec:committee_example}

Consider a committee of $n=6$ members deciding whether a project is viable. There are various roles in the committee, such as the engineer, the economist, and the manager. Each one may learn and aggregate information from others differently. During the decision process, one can expect a first stage consisting of an individual assessment of viability, so each role can decide how she wants to aggregate the information the others are providing. However, after some research, the committee must meet and decide the verdict. Then, during the meeting process, everyone should aggregate others' information in the same way, although some agents may be more influential than others and their opinions are more important. During both stages, every member keeps receiving private evidence in each period.

Following this idea, we compare four scenarios in which members use either homogeneous or heterogeneous rules and either meet or do not meet, so that each ingredient of the model can be assessed separately.

\subsubsection{Experimental setup}

We define $K=3$ hypotheses. The project is either viable ($\theta_V=\tstar$, which is the true state), not viable for technological reasons ($\theta_T$), or not viable for cost reasons ($\theta_C$). The committee is formed by six individuals, two for each of the mentioned roles. In every period each member sees a binary signal (favorable or unfavorable), independent across members and i.i.d.\ over time, coming from the true state distribution $l^{\theta^*}_i$. Table~\ref{tab:committee_signals} gives the signal structure and the relative entropies $h_i(\hat\theta)$. Engineers can tell $\theta_V$ from $\theta_T$ but cannot see cost at all ($h_i(\theta_C)=0$), economists are the mirror image, while managers can simply identify whether it is viable or not. Only the manager can identify $\tstar$ alone and every false state can be told apart from $\tstar$ by someone, so Hypothesis~\ref{hyp:H3} holds. We set each individual's prior beliefs as uniform for each hypothesis (Hypothesis~\ref{hyp:H4}).

\begin{table}[ht!]
    \centering
    \begin{tabular}{lccccccc}
        \toprule
        Role & $\ell_i(\text{fav}\mid\theta_V)$ & $\ell_i(\text{fav}\mid\theta_T)$ & $\ell_i(\text{fav}\mid\theta_C)$ & $h_i(\theta_T)$ & $h_i(\theta_C)$ & $\rho_{\text{role}}$ \\
        \midrule
        Engineer  & 0.70 & 0.30 & 0.70 & 0.339 & 0     & $-1$ \\
        Economist & 0.70 & 0.70 & 0.30 & 0     & 0.339 & $0$  \\
        Manager   & 0.55 & 0.45 & 0.45 & 0.020 & 0.020 & $+1$ \\
        \bottomrule
    \end{tabular}
    \caption{Signal structure, informativeness, and aggregation parameters by role in the committee example. For each role, $\ell_i(\text{fav}\mid\theta)$ is the probability that a member observes a favorable signal when the true state is $\theta$, where $\theta_V$ denotes a viable project, $\theta_T$ a project that is not viable for technological reasons, and $\theta_C$ a project that is not viable for cost reasons. The relative entropies $h_i(\theta_T)$ and $h_i(\theta_C)$ measure how well a member of each role distinguishes the true state $\theta_V$ from $\theta_T$ and $\theta_C$, respectively. The parameter $\rho_{\text{role}}$ is the CES aggregation parameter used by each role in the heterogeneous scenarios.}
    \label{tab:committee_signals}
\end{table}

The graph formed by the agents when they communicate is complete. We fix $\aund=1/n$, $\bar a=(n-1)/n=5/6$ and $\lambda_t=\max\{t,1.01\}^{-0.3}$. The value $1.01$ is used at $t\in\{0,1\}$, so that it satisfies the definition of $\lambda_t\in(0,1)$. Then, Hypothesis \ref{hyp:H2} holds since $t\lambda_t=t^{0.7}\to\infty$. We then define the weights of the adjacency matrix as:

\begin{equation*}
    a_{ij,t+1}=\frac{\lambda_t}{n}\ (j\neq i),\qquad a_{ii,t+1}=1-\frac{n-1}{n}\lambda_t ,
\end{equation*}

so that $b_{ij,t}=a_{ij,t+1}$ for all $i,j$, as in Section~\ref{subsec:CES}.

At $t=0$ each member weights herself and every colleague equally, at $1/n$. As $\lambda_t$ decays, each member stops listening to others and the weight drifts toward herself. For example, at $t=10$ the weight assigned to her own belief is $a_{ii,10}=0.57$, and the remaining attention is spread across her neighbors.

\subsubsection{Scenarios}

As explained, we can divide the decision process into two regimes if the agents meet at a specified time period, set at $t_c=10$. In every scenario, every aggregation rule $\psi_{i,t}$ used belongs to the CES family, $\psi_{i,t}(x)=\bigl(\sum_j b_{ij,t}x_j^{\rho_{i,t}}\bigr)^{1/\rho_{i,t}}$, so Hypothesis \ref{hyp:H1} holds. In Table~\ref{tab:scenarios} we show the four scenarios we compare, all covered by our theory, which covers homogeneous and heterogeneous aggregation rules between agents and the possibility of changing that rule over time in a non-stochastic way. In the scenarios where rules are homogeneous, we specify that agents use the log-linear rule of \cite{molavi2018theory} ($\rho_{i,t}\equiv0$), and if they meet, it changes at $t_c$ to the harmonic rule ($\rho_{i,t\geq t_c} = \bar{\rho} = -1$). In the heterogeneous cases, each role aggregates information differently:  the engineer aggregates following the harmonic rule ($\rho_{\text{eng}}=-1$), the economist uses the log-linear ($\rho_{\text{econ}}=0$), while the manager aggregates linearly ($\rho_{\text{manag}}=1$). If they meet, their aggregation rules will switch to the harmonic rule ($\bar{\rho}=-1$). For each scenario, we run $300$ simulations with different signal realizations.

\begin{table}[ht!]
    \centering
    \begin{tabular}{lcc}
        \toprule
         & Fixed rules & Meeting at $t_c=10$ \\
        \midrule
        Homogeneous   & $\rho_{i,t}\equiv 0$ & $0$ for $t<t_c$, $\bar\rho=-1$ for $t\ge t_c$ \\
        Heterogeneous & $\rho_{i,t}\equiv\rho_{\text{role}(i)}$ & $\rho_{\text{role}(i)}$ for $t<t_c$, $\bar\rho=-1$ for $t\ge t_c$ \\
        \bottomrule
    \end{tabular}
    \caption{Description of the four scenarios run, where we fix either heterogeneous or homogeneous aggregation rules and if the committee meets or not.}
    \label{tab:scenarios}
\end{table}

\subsubsection{Results}

We show in Figure~\ref{fig:belief_evolution} the evolution over time of the mean belief in $\tstar$, over agents and over the different runs, for the different scenarios considered, as well as the $10-90\%$ quantile band. In Figure \ref{fig:remaining_doubt}, we show the remaining doubt in log scale, computed as the logarithm of the sum of beliefs assigned to hypotheses other than the true state, i.e. $\log(1-\mu_{i,t}(\tstar))$. 

\begin{figure}[ht!]
    \centering
    \begin{subfigure}[t]{0.49\textwidth}
        \centering
        \includegraphics[width=\textwidth]{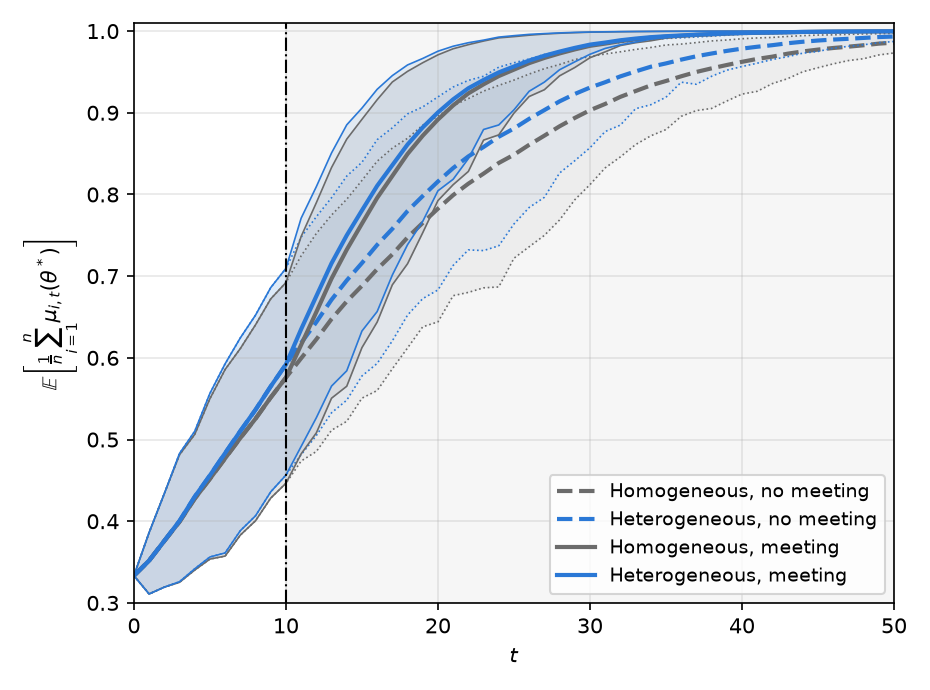}
        \caption{Belief evolution}
        \label{fig:belief_evolution}
    \end{subfigure}
    \hfill
    \begin{subfigure}[t]{0.49\textwidth}
        \centering
        \includegraphics[width=\textwidth]{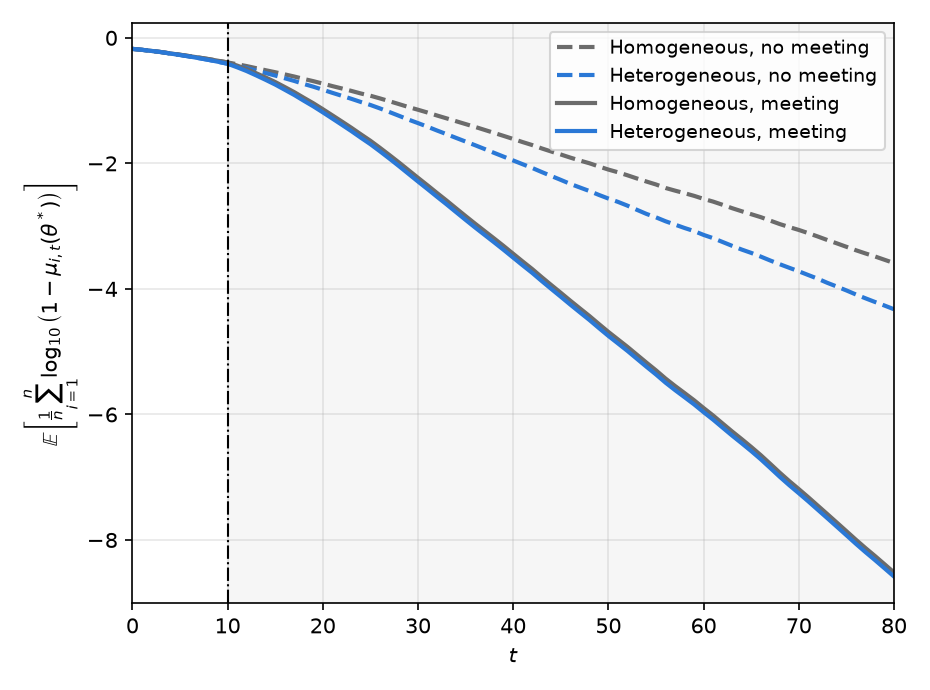}
        \caption{Remaining doubt}
        \label{fig:remaining_doubt}
    \end{subfigure}

    \caption{Evolution of beliefs and remaining doubt over time in the committee example. In gray, we show the homogeneous aggregation rules scenarios, in blue the heterogeneous ones. Dashed lines consider the scenario where there is no meeting, while solid means a meeting starts at $t_c=10$ indicated by the dash-dotted vertical line. (a) shows the belief in the true state averaged over members and across the $300$ runs, where bands are the $10-90\%$ quantile bands (delimited with tighter lines). (b) shows the remaining doubt (log of the belief assigned to other hypotheses $\log(1-\mu_{i,t}(\tstar))$) averaged over members and runs.}
    \label{fig:committee}
\end{figure}

In all four scenarios the Hypotheses presented in the article hold, and the switch of the aggregation rules at $t_c$ is deterministic. Theorem~\ref{thm:main} therefore gives that all the members beliefs in the committee will converge almost surely to $\mu_{i,t}(\tstar) \to 1 \ , \ \forall i\in n$. Figure \ref{fig:belief_evolution} is consistent with this, since the mean belief goes from the uniform prior $\frac{1}{3}$ towards $1$ in every scenario, and the quantile bands narrow around $1$ as $t$ grows, and never approach $0$. Figure \ref{fig:remaining_doubt} shows the same, since the log remaining doubt keeps decreasing as time increases. What separates the scenarios is how fast this happens, and that is where the CES-specific results of Section \ref{subsec:CES} apply. 

When analyzing the effects of the meeting, i.e. the change to the common harmonic aggregation rule, we see that after $t_c$ the belief curves for the meeting scenarios lie above those without a meeting. Moreover, their doubt is lower by a margin that grows with $t$. This holds in both the homogeneous and heterogeneous cases. The explanation lies in the sign of $\rho$. After $t_c$ every member uses the harmonic rule on a complete graph, and Corollary \ref{cor:rateneg} says that under negative curvature, each member follows the best-informed one, i.e. the engineers for identifying $\theta_T$ and the economists for $\theta_C$. Moreover, recall that $\rho=-1$ gives the best bound among the negative values for a given $\lambda_t$. The meeting acts as an institutional device that, without providing any new evidence, changes how information is pooled and lets each member benefit from the expertise of others on the dimension where they are competent.

If instead we focus on heterogeneity, we see that when the committee meets, the two meeting curves are nearly indistinguishable after $t_c$, since the common rule removes the pre-meeting differences between the members, and we set the adjacency matrix weights equal across individuals, i.e. no role is more influential during the decision process. In contrast, when they do not meet, we observe in Figure \ref{fig:remaining_doubt} that the remaining doubt under heterogeneity decreases faster over time than in the homogeneous case. This is because the homogeneous log-linear rule ($\rho=0$) lies between the harmonic and the linear rules. Its convergence rate bound benchmark, shown in \cite{molavi2018theory}, is the average informativeness of the committee, which is much lower than that of the best-informed member given by the harmonic rule because the managers are nearly uninformative. Therefore, the presence of agents with the harmonic rule in the heterogeneous case increases the speed of convergence.

Heterogeneous and time-varying aggregation does not affect almost sure convergence (Theorem~\ref{thm:main}), but it does affect the speed of convergence, i.e. how quickly doubt disappears. As explained in CES family Section \ref{subsec:CES} and suggested by Section \ref{sec:dominating}, rules with more negative curvature let members rely on the most informative colleagues, and rules with positive curvature are slowed by the least informed members.

\section{Discussion and Future Work}\label{sec:discus_and_future}
Our framework enables us to study a broader set of aggregation functions, as well as heterogeneous, time-varying rules. However, there is still room for improvement, and several questions remain open. 

\paragraph{Correlated observations.} Hypothesis \ref{hyp:H3} states independence of the received signals across agents. Even though it is a standard assumption, it is still a strong Hypothesis, which is used only once in the proof. This is why we believe that it is possible to obtain Theorem \ref{thm:main} without this supposition. By contrast, we do not require this constraint for Theorems \ref{thm:lower} and \ref{thm:upper}, making them applicable to most settings, especially more realistic ones. Indeed, in practice, the independence assumption across agents is very unrealistic, because, often, the signals originate from a single cause, and therefore, are very likely to be dependent.

\paragraph{Tightness of the convergence rate bounds.} One can wonder whether this bound is tight, we were not able to prove it. To arrive at Theorem \ref{thm:upperbound}, we applied a straightforward upper bounding on the equation (i) obtained in Theorem \ref{thm:linearrep}, therefore, we make the conjecture that this bound might be improved. Concerning the bounds obtained by Theorems \ref{thm:lower} and \ref{thm:upper}, In particular, in the case of Theorem \ref{thm:upper}, when the bound is equal to exactly $0$ (existence of a blind agent and $\lambda_t$ decaying to 0, then the bound is tight) for Theorem $\ref{thm:lower}$, in some cases, where there is a concentration of very informed agents in the graph, the bound might be surpassed. Moreover, Proposition \ref{prop:S=R} gives an exact rate, but holds for very specific combinations of graphs and aggregators. 

\paragraph{Misspecification.} We placed ourselves in a very convenient setting where there is a true unknown state $\tstar\in\Theta$, with $\Theta$ being a finite set known by the agents. Furthermore, they know the likelihoods linking received signals to these states, where the distribution of the signals they observe is exactly the same as the likelihood on the true state. Do these results hold when this equality no longer holds but the likelihood under the true state remains the closest to the real signal distribution? 

\paragraph{Time-varying neighborhoods.} Another strong constraint on our framework is that neighborhoods are fixed for all agents. This is less restrictive than the other points, however, it could be useful to allow agents to change the set of people with whom we communicate.

\paragraph{Optimal configuration under heterogeneous time varying rules.} As we have shown in Section \ref{subsec:CES}, from our setting we conjecture that any network should use the harmonic rule with weights decaying to 0 (having the best lower bound) , but this is only proved for the CES family case and under homogeneity among agents. Would this remain true in a more realistic and complex setup, such as heterogeneous families of aggregation functions?

\section{Conclusion}\label{sec:conclusions}

We have studied non-Bayesian social learning when agents combine the beliefs of their neighbors through heterogeneous aggregation rules that may change over time in a predetermined way. The framework rests on two primitives of each rule. First, the trust weights describe how much an individual listens to each of her neighbors. Second, the logarithmic curvature describes how she listens to others, from deferring to the most convinced voices ($\kappa>0$) to deferring to the most skeptical ones ($\kappa<0$). Under bounded curvature, strong connectivity, a sufficiently slow decay of attention to neighbors, and sufficient capacity among the agents in the network to distinguish the true state from the incorrect ones, every agent learns the true state almost surely (Theorem~\ref{thm:main}). Thus, heterogeneity and time variation in how people weigh each other's opinions do not prevent a society from aggregating its dispersed information. Learning can fail if individuals can be captured by extreme opinions or if they stop listening to others too early, which can also affect those who listen to them.

We showed that any admissible rule can be written exactly as a log-linear rule with stochastic, belief-dependent weights (Theorem~\ref{thm:linearrep}). Behind any norm of opinion pooling, then, lies an effective structure of influence that emerges from the current distribution of opinions and the actual aggregation functions. This representation yields a general upper bound on the learning rate that does not depend on the aggregation rules (Theorem~\ref{thm:upperbound}). In other words, communication can bring a group closer to this bound, but only better information and the ability to distinguish states.

By defining dominating and dominated paths, linked to the sign of the curvature (Proposition~\ref{prop:curv}), we obtained individual bounds on the convergence rate that do not require independent signals across agents (Theorems~\ref{thm:lower} and~\ref{thm:upper}). An agent learns at least as fast as the best-informed agent she reaches through a chain of followership, and no faster than the least-informed agent she cannot discount, up to the addition of the term proportional to $\beta_\lambda$. How fast an individual learns thus depends less on her own information than on her position in the chains of influence of her society. For the CES family, negative curvature makes every agent learn at least at the rate of the best-informed agent when $\beta_\lambda=0$, with the harmonic rule giving the best lower bound (Corollary~\ref{cor:rateneg}). Positive curvature lets the least-informed agent bound the rate of all others, so that a single uninformed member can slow the whole network to a sub-exponential rate (Corollary~\ref{cor:ratepos}). Two societies with identical networks and information may therefore remain mistaken for very different lengths of time, solely because of how they decide to weight each other's opinions. The committee example shows that a deliberation stage that changes how opinions are pooled, without adding evidence, can shorten this period.

The main limitations and open directions for future work are as follows. Almost sure convergence relies on independent signals across agents, which common sources of information often violate. Agents are assumed correctly specified, whereas norms deferring to the least informed could amplify the errors of a misinformed or manipulative individual. Neighborhoods in our model are fixed, which is not realistic in dynamic environments. Finally, the tightness of our bounds and the optimal configuration of heterogeneous rules beyond the CES family remain open, and their answers would inform the design of deliberation procedures.

\newpage

\appendix

\section{Appendix}\label{app:one}

In Appendix \ref{app:one} we prove Theorem~\ref{thm:main}.

We denote by $\mathcal{F}_t:=\sigma\bigl(\omega_{i,s}: i\in\cN,\ s\le t\bigr)$ the filtration generated by $\mathbb{P}^{\theta^*}$ at time $t$, and by $\mathbb{E}_t[\,\cdot\,]:=\mathbb{E}^{\tstar}[\,\cdot\mid\mathcal{F}_t]$ the conditional expectation with respect to that filtration. The proof is composed of two parts. In the first one, we prove the following Proposition.

\begin{proposition}[Vanishing distinguishable interim beliefs]\label{prop:vanish}
For every agent $i$ and every $\theta\in\Theta_i$,

\begin{equation*}
    \sum_{t=0}^{\infty}\zeta_{it}(\theta)^2<\infty
    \quad\Pstar\text{-a.s.},
    \qquad\text{and in particular}\qquad
    \zeta_{it}(\theta)\xrightarrow[t\to\infty]{}0
    \quad\Pstar\text{-a.s.}
\end{equation*}
\end{proposition}

In the second part, we show the propagation of this result through the communication graph and we prove Theorem~\ref{thm:main}.

We define the following terms:

\begin{equation}\label{eq:deltaeps}
    \delta_{it}(\theta):=\sum_{j=1}^n a_{ij,t+1}\,\mu_{jt}(\theta)-\psi_{i,t}\bigl(\mu_t(\theta)\bigr),
    \qquad
    \epsilon_{it}(\theta):=\zeta_{it}(\theta)-\sum_{j=1}^n a_{ij,t+1}\,\mu_{jt}(\theta),
\end{equation}

and we set $c_\ell:=\min_{i,\omega,\theta}\ell^{\theta}_i(\omega)>0$.

Moreover, we use Lemma A.1 of \citet{molavi2018theory}: under Hypotheses~\ref{hyp:H1} and~\ref{hyp:H2}, there exists a sequence of probability vectors $(v_t)_{t\ge1}$, $v_t=(v_{1,t},\dots,v_{n,t})$ with $v_{i,t}>0$ and $\sum_{i=1}^n v_{i,t}=1$, satisfying the left-eigenvector identity

\begin{equation}\label{eq:tracking}
    \sum_{i=1}^{n} v_{i,t+2}\,a_{ij,t+1}=v_{j,t+1}
    \qquad\text{for all } j\in\cN,\ t\ge 0,
\end{equation} 

and admitting a uniform floor $v_{i,t}\ge c_v>0$ for all $i,t$. We call these weights \emph{centrality weights}.

\begin{definition}\label{def:nuphi}
We define

\begin{equation*}
    \nu_{it}:=\frac{1}{\mu_{it}(\theta^*)}
    \qquad\text{and}\qquad
    \phi_{it}(x):=\frac{1}{\psi_{i,t}\bigl(1/x_1,\dots,1/x_n\bigr)}.
\end{equation*}
\end{definition}

\begin{lemma}\label{lem:phi}
The function $\phi_{it}$ is concave and homogeneous of degree $1$, and its logarithmic curvature is less than or equal to $1$. Moreover,

\begin{equation*}
    \frac{\partial \phi_{it}}{\partial x_j}(1,\dots,1)=a_{ij,t+1}
    \qquad\text{and}\qquad
    \phi_{it}(x)\le\sum_{j=1}^{n}a_{ij,t+1}\,x_j.
\end{equation*}
\end{lemma}

\begin{proof}
Using Lemmas~A.9 and~A.10 of \citet{molavi2018theory}, we obtain that $\phi_{it}$ is concave and homogeneous of degree $1$, with a logarithmic curvature that is less than or equal to $1$, and that

\begin{equation*}
    \phi_{it}(x)\le\sum_{j=1}^{n}\phi_{it}^{(j)}(1,\dots,1)\,x_j, \qquad \phi_{it}^{(j)}=\frac{\partial \phi_{it}}{\partial x_j}.
\end{equation*}

For all $x\in (0,\infty)^n$, one has $\phi_{it}(x)=1/\psi_{i,t}(y)$, with $y_j=1/x_j$. The chain rule gives

\begin{equation*}
    \frac{\partial \phi_{it}}{\partial x_k}(x) =-\frac{1}{\psi_{i,t}(y)^2}\,\frac{\partial \psi_{i,t}(y)}{\partial x_k}.
\end{equation*}

Since

\begin{equation*}
    \frac{\partial y_j}{\partial x_k}=
    \begin{cases}-\dfrac{1}{x_k^2} & \text{if } j=k,\\[1ex] 0 & \text{if } j\neq k,\end{cases}
    \qquad\text{then}\qquad
    \frac{\partial \psi_{i,t}(y)}{\partial x_k}=-\frac{1}{x_k^2}\,\frac{\partial \psi_{i,t}}{\partial y_k}(y).
\end{equation*}

Then

\begin{equation*}
    \frac{\partial \phi_{it}}{\partial x_k}(x) =\frac{1}{x_k^2\,\psi_{i,t}(1/x_1,\dots,1/x_n)^2}\,\frac{\partial \psi_{i,t}}{\partial y_k}(y).
\end{equation*}

Evaluated at $(1,\dots,1)$, this gives

\begin{equation*}
    \frac{\partial \phi_{it}}{\partial x_k}(1,\dots,1)=\frac{\partial \psi_{i,t}}{\partial y_k}(1,\dots,1)=a_{ik,t+1},
\end{equation*}
and therefore

\begin{equation*}
    \phi_{it}(x)\le\sum_{j}a_{ij,t+1}x_j. \qedhere
\end{equation*}
\end{proof}

\begin{lemma}\label{lem:delta}
For every $\hat\theta\in\Theta$ we have $\delta_{it}(\hat\theta)\ge0$. Moreover,

\begin{equation*}
    \sum_{\tilde\theta\in\Theta}\delta_{it}(\tilde\theta)=1-S_{it}, \qquad\text{so that}\qquad 0\le\delta_{it}(\tilde\theta)\le1-S_{it},
\end{equation*}

and $\bigl|\epsilon_{it}(\tilde\theta)\bigr|\le\delta_{it}(\tilde\theta)+(1-S_{it})$.
\end{lemma}

\begin{proof}
We simply use Lemma~A.9 of \citet{molavi2018theory} applied to $\psi_{i,t}$:

\begin{equation*}
    \psi_{i,t}\bigl(\mu_t(\theta)\bigr)\le\sum_{j}a_{ij,t+1}\mu_{jt}(\theta), \qquad\text{so}\qquad \delta_{it}(\theta)\ge0.
\end{equation*}

Summing the definition \eqref{eq:deltaeps} of $\delta_{it}$ over $\tilde\theta\in\Theta$ and using $\sum_{\tilde\theta}\mu_{jt}(\tilde\theta)=1$ together with $\sum_ja_{ij,t+1}=1$ gives $\sum_{\tilde\theta}\delta_{it}(\tilde\theta)=1-S_{it}$; since each term is non-negative, each is bounded by the sum. The last inequality follows from the definition of $\epsilon_{it}$ in \eqref{eq:deltaeps}.
\end{proof}

\begin{lemma}\label{lem:onestep}
For every $i\in\mathcal N$, one has $S_{i,t}\le1$  and

\begin{equation*}
    \mathbb E_t\bigl[\nu_{i,t+1}\bigr]
    =\frac{1}{\zeta_{it}(\theta^*)}
    \le\sum_{j=1}^{n}a_{ij,t+1}\,\nu_{jt}.
\end{equation*}
\end{lemma}

\begin{proof}
\begin{equation*}
    \mathbb E_t(\nu_{i,t+1})
    =\mathbb E_t\!\left[\frac{m_{it}(\omega_{i,t+1})}{\zeta_{it}(\theta^*)\,\ell_i^{\theta^*}(\omega_{i,t+1})}\right]
    =\sum_{\omega\in\Omega_i}\frac{m_{it}(\omega)}{\zeta_{it}(\theta^*)}
    =\frac{1}{\zeta_{it}(\theta^*)}\sum_{\theta\in\Theta}\sum_{\omega\in\Omega_i}\zeta_{it}(\theta)\ell_i^{\theta}(\omega)
    =\frac{S_{it}}{\psi_{i,t}(\mu_t(\theta^*))}.
\end{equation*}

Moreover, by Lemma~\ref{lem:delta}, it is immediate that $S_{i,t}\le1$

Therefore, by Definition~\ref{def:nuphi} and then Lemma~\ref{lem:phi},

\begin{equation*}
    \mathbb E_t\bigl[\nu_{i,t+1}\bigr]\le\phi_{it}(\nu_t)\le\sum_{j=1}^{n}a_{ij,t+1}\nu_{jt}. \qedhere
\end{equation*}
\end{proof}

We then define

\begin{equation*}
    Y_t:=\sum_{i=1}^{n}v_{i,t+1}\,\nu_{it}\ \ge\ 0.
\end{equation*}

\begin{proposition}\label{prop:smg}
The process $(Y_t)_{t\ge0}$ is a $\mathbb{P}^{\theta^*}$-supermartingale, positive, and therefore converges $\mathbb{P}^{\theta^*}$-almost surely to a random variable which we denote by $Y_\infty$. Moreover, there exists a random variable $C>0$ such that, for all $i$ and $t$,

\begin{equation*}
    \mu_{it}(\theta^*)\ge C>0\qquad\mathbb{P}^{\theta^*}\text{-almost surely.}
\end{equation*}
\end{proposition}

\begin{proof}
By Lemma~\ref{lem:onestep} and then by the identity \eqref{eq:tracking},

\begin{equation*}
    \mathbb E_t\bigl[Y_{t+1}\bigr] \le\sum_i v_{i,t+2}\Bigl(\sum_{j=1}^{n}a_{ij,t+1}\nu_{jt}\Bigr) =\sum_{j=1}^{n}v_{j,t+1}\nu_{jt}=Y_t.
\end{equation*}

Moreover, $Y_t$ is $\mathcal{F}_t$-measurable, and $\mathbb E^{\theta^*}[Y_t]\le\mathbb E^{\theta^*}[Y_0]=Y_0<\infty$ by Hypothesis~\ref{hyp:H4}. Then $(Y_t)_{t\ge0}$ is a non-negative supermartingale, hence converges to $Y_\infty$ $\mathbb{P}^{\theta^*}$-almost surely. Moreover, using the uniform floor $v_{i,t}\ge c_v$,

\begin{equation*}
    \nu_{it}\le\frac{Y_t}{c_v}\le\frac{\sup_s Y_s}{c_v}<\infty\qquad\mathbb{P}^{\theta^*}\text{-a.s.},
\end{equation*}

so that

\begin{equation*}
    \mu_{it}(\theta^*)\ge\frac{c_v}{\sup_s Y_s}=:C>0\qquad\mathbb{P}^{\theta^*}\text{-a.s.} \qedhere
\end{equation*}
\end{proof}

The random variable $C$ is defined on the set of sequences of signals, like $\sup_sY_s$.

\begin{lemma}\label{lem:drift}
Define

\begin{equation*}
    \Delta_{1,it}:=\sum_{j}a_{ij,t+1}\nu_{jt}-\frac{1}{\psi_{i,t}(\mu_t(\theta^*))}\ \ge\ 0, \qquad \Delta_{2,it}:=\frac{1-S_{it}}{\psi_{i,t}(\mu_t(\theta^*))}\ \ge\ 0.
\end{equation*}

Then

\begin{equation}\label{eq:Ddecomp}
    D_t:=Y_t-\mathbb E_t[Y_{t+1}]=\sum_{i=1}^{n}v_{i,t+2}\bigl(\Delta_{1,it}+\Delta_{2,it}\bigr).
\end{equation}
\end{lemma}

\begin{proof}
Euler's inequality, as given by Lemma~\ref{lem:phi}, implies $\Delta_{1,it}\ge0$. By Lemma~\ref{lem:onestep} we have $S_{it}\le1$, and $\psi_{i,t}\ge0$, so $\Delta_{2,it}\ge0$. One has

\begin{equation*}
    Y_t-\mathbb E_t[Y_{t+1}] =\sum_{j}v_{j,t+1}\nu_{jt}-\sum_i v_{i,t+2}\,\mathbb E_t[\nu_{i,t+1}]
\end{equation*}

\begin{equation*}
    =\sum_j\Bigl(\sum_{i}v_{i,t+2}a_{ij,t+1}\Bigr)\nu_{jt}-\sum_i v_{i,t+2}\frac{S_{it}}{\psi_{i,t}(\mu_t(\theta^*))}
\end{equation*}

\begin{equation*}
    =\sum_i v_{i,t+2}\Bigl(\sum_{j}a_{ij,t+1}\nu_{jt}-\frac{S_{it}}{\psi_{i,t}(\mu_t(\theta^*))}\Bigr) =\sum_{i=1}^{n}v_{i,t+2}\bigl(\Delta_{1,it}+\Delta_{2,it}\bigr),
\end{equation*}

where the second equality uses the identity \eqref{eq:tracking} together with Lemma~\ref{lem:onestep}.
\end{proof}

\begin{lemma}\label{lem:sumD}
\begin{equation*}
    \sum_{t=0}^{\infty}\mathbb E^{\theta^*}[D_t]\le Y_0<+\infty.
\end{equation*}
\end{lemma}

\begin{proof}
By Lemma~\ref{lem:drift}, $D_t=Y_t-\mathbb E_t[Y_{t+1}]$, hence $\mathbb E^{\theta^*}[D_t]=\mathbb E^{\theta^*}[Y_t]-\mathbb E^{\theta^*}[Y_{t+1}]$. For $N\in\mathbb N$,

\begin{equation*}
    \sum_{t=0}^{N-1}\mathbb E^{\theta^*}[D_t]=\mathbb E^{\theta^*}[Y_0]-\mathbb E^{\theta^*}[Y_N]\le Y_0.
\end{equation*}

Since $D_t\ge0$ by Lemma~\ref{lem:drift}, we obtain the result by letting $N$ tend to infinity.
\end{proof}

\begin{lemma}[Summable expected remainders]\label{lem:remainders}
For every agent $i$ and every state $\theta\in\Theta$,

\begin{align}
    \sum_{t=0}^{\infty}\Estar[1-S_{it}]&\;\le\;\frac{Y_0}{c_v}, \label{eq:sumS}\\
    \sum_{t=0}^{\infty}\Estar[\delta_{it}(\theta)]&\;\le\;\frac{Y_0}{c_v}, \label{eq:sumdelta}\\
    \sum_{t=0}^{\infty}\Estar\bigl[|\epsilon_{it}(\theta)|\bigr]&\;\le\;\frac{Y_0}{c_v}. \label{eq:sumeps}
\end{align}
\end{lemma}

\begin{proof}
The idea is the following: each of these remainders is dominated, pointwise and up to the constant $c_v$, by $D_t=Y_t-\mathbb E_t[Y_{t+1}]$ whose expectation is summable with $\sum_{t\ge0}\Estar[D_t]\le Y_0$ by Lemma~\ref{lem:sumD}. It therefore suffices to establish the pointwise domination in each case. In the following, fix $\theta\in\Theta$

\begin{itemize}
    \item {0: $\phi_{i,t}(\nu_t)\ge1$.} By Definition~\ref{def:nuphi},
    
    \begin{equation*}
        \phi_{i,t}(\nu_t) =\frac{1}{\psi_{i,t}\bigl(1/\nu_{1t},\dots,1/\nu_{nt}\bigr)} =\frac{1}{\psi_{i,t}\bigl(\mu_t(\theta^*)\bigr)},
    \end{equation*}
    
    since $\nu_{jt}=1/\mu_{jt}(\theta^*)$ by construction. Now the beliefs are probabilities, so $\mu_{jt}(\theta^*)\le1$ for every $j$, as $\psi_{i,t}$ is weakly increasing in each argument and $\psi_{i,t}(\mathbf 1)=1$ (Hypothesis~\ref{hyp:H1}), we obtain
    
    \begin{equation*}
        \psi_{i,t}\bigl(\mu_t(\theta^*)\bigr)\;\le\;\psi_{i,t}(\mathbf 1)\;=\;1, \qquad\text{hence}\qquad \phi_{i,t}(\nu_t)\;\ge\;1 .
    \end{equation*}
    
    Moreover, $\phi_{i,t}$ is itself weakly increasing with $\phi_{i,t}(\mathbf 1)=1$, and $\nu_t\ge\mathbf 1$ componentwise.
    
    \item {1: bound \eqref{eq:sumS}.} One has $\Delta_{2,it}=(1-S_{it})\,\phi_{i,t}(\nu_t)$. Since $S_{it}\le1$ by Lemma~\ref{lem:onestep}, the factor $1-S_{it}$ is non-negative, so point~0 gives
    
    \begin{equation*}
        1-S_{it}\;\le\;(1-S_{it})\,\phi_{i,t}(\nu_t)\;=\;\Delta_{2,it}.
    \end{equation*}
    
    In the decomposition \eqref{eq:Ddecomp} every term in the sum is positive, then, only keeping the one indexed by $i$ and using $v_{i,t+2}\ge c_v$ yields
    
    \begin{equation*}
        c_v\,(1-S_{it})\;\le\;c_v\,\Delta_{2,it}\;\le\;v_{i,t+2}\,\Delta_{2,it}\;\le\;D_t .
    \end{equation*}
    
    Taking expectations and summing over $t$, Lemma~\ref{lem:sumD} gives $\sum_t\Estar[1-S_{it}]\le Y_0/c_v$.
    
    \item{2: bound \eqref{eq:sumdelta}.} By Lemma~\ref{lem:delta}, $\delta_{it}(\theta)$ is non-negative and the total sums to $\sum_{\theta\in\Theta}\delta_{it}(\theta)=1-S_{it}$. A single non-negative term of a sum is bounded by the sum, so $0\le\delta_{it}(\theta)\le1-S_{it}$ for each $\theta$, and \eqref{eq:sumdelta} follows from \eqref{eq:sumS}.
    
    \item{3: bound \eqref{eq:sumeps}.} By \eqref{eq:deltaeps} and \eqref{eq:interim} one has the identity
    
    \begin{equation*}
        \epsilon_{it}(\theta)=\zeta_{it}(\theta)\,(1-S_{it})-\delta_{it}(\theta),
    \end{equation*}
    
    obtained by writing $\sum_j a_{ij,t+1}\mu_{jt}(\theta)=\psi_{i,t}(\mu_t(\theta))+\delta_{it}(\theta)$ and $\psi_{i,t}(\mu_t(\theta))=\zeta_{it}(\theta)S_{it}$. The two terms on the right have opposite signs and are both bounded by $1-S_{it}$: indeed $0\le\zeta_{it}(\theta)\le1$ because $\zeta_{it}$ is a probability distribution, and $0\le\delta_{it}(\theta)\le1-S_{it}$. Hence $-(1-S_{it})\le\epsilon_{it}(\theta)\le1-S_{it}$, that is $|\epsilon_{it}(\theta)|\le1-S_{it}$, and \eqref{eq:sumeps} again follows from \eqref{eq:sumS}.
\end{itemize}
\end{proof}

We denote by $\chi_i^2(t)$ the Pearson divergence between $m_{i,t}$ and $l_i^{\theta^*}$:

\begin{equation}\label{eq:chi2}
    \chi^2_i(t):=\sum_{\omega\in\Omega_i}\frac{\bigl(\ell^{\theta^*}_i(\omega)-m_{it}(\omega)\bigr)^2}{m_{it}(\omega)} .
\end{equation}

It appears naturally in the next Proposition.

\begin{proposition}\label{prop:chi}
Denote

\begin{equation*}
    X_t:=\sum_{i=1}^{n}v_{i,t+1}\mu_{it}(\theta^*),
    \qquad
    V_t:=1-X_t\ge0.
\end{equation*}

Then $V_t\to V_\infty$ $\mathbb{P}^{\theta^*}$-a.s.\ and

\begin{equation*}
    \sum_{t=0}^{\infty}\mathbb E^{\theta^*}\!\left[\sum_i \zeta_{it}(\theta^*)\chi_i^2(t)\right]<+\infty.
\end{equation*}
\end{proposition}

\begin{proof}
First, one has

\begin{equation*}
    \Et\bigl[\mu_{i,t+1}(\tstar)\bigr] =\sum_{\omega\in\Omega_i}\ell^{\tstar}_i(\omega)\,\zeta_{it}(\tstar)\,\frac{\ell^{\tstar}_i(\omega)}{m_{it}(\omega)} =\zeta_{it}(\tstar)\sum_{\omega\in\Omega_i}\frac{\ell^{\tstar}_i(\omega)^2}{m_{it}(\omega)} .
\end{equation*}

Adding and subtracting $m_{it}(\omega)$ in the numerator,

\begin{equation*}
    \frac{\ell^{\tstar}_i(\omega)^2}{m_{it}(\omega)} =\frac{(\ell^{\tstar}_i(\omega)-m_{it}(\omega))^2}{m_{it}(\omega)} +2\bigl(\ell^{\tstar}_i(\omega)-m_{it}(\omega)\bigr)+m_{it}(\omega),
\end{equation*}

and summing over $\omega$ (both $\ell^{\tstar}_i$ and $m_{it}$ sum to one), we obtain by \eqref{eq:chi2}

\begin{equation}\label{eq:exactdrift}
    \Et\bigl[\mu_{i,t+1}(\tstar)\bigr]
    =\zeta_{it}(\tstar)\bigl(1+\chi^2_i(t)\bigr).
\end{equation}

Then

\begin{equation*}
    \mathbb E_t\bigl(X_{t+1}\bigr) =\sum_{i=1}^{n}v_{i,t+2}\,\mathbb E_t\bigl(\mu_{i,t+1}(\theta^*)\bigr) =\sum_i v_{i,t+2}\zeta_{it}(\theta^*)+\sum_i v_{i,t+2}\zeta_{it}(\theta^*)\chi_i^2(t),
\end{equation*}

so that

\begin{equation*}
    \mathbb{E}_t\bigl(X_{t+1}\bigr)=X_t-\sum_{i=1}^{n}v_{i,t+2}\Bigl[\delta_{it}(\theta^*)\Bigl(1+\frac{1-S_{i,t}}{S_{i,t}}\Bigr) -\frac{1-S_{i,t}}{S_{i,t}}\sum_{j}a_{ij,t+1}\mu_{jt}(\theta^*)\Bigr]
\end{equation*}

\begin{equation*}
    \qquad\qquad+\sum_{i}v_{i,t+2}\zeta_{it}(\theta^*)\chi_i^2(t).
\end{equation*}

Since $X_t$ takes values between $0$ and $1$, so does $V_t$. Define $\gamma_t:=\sum_{i\ge1}v_{i,t+2}\frac{\delta_{it}(\theta^*)}{S_{it}}\ge0$, one then obtains,

\begin{equation*}
    0\le\mathbb E_t\bigl(V_{t+1}\bigr)\le V_t+\gamma_t-\sum_{i=1}^{n}v_{i,t+2}\zeta_{it}(\theta^*)\chi_i^2(t).
\end{equation*}

Using $\delta_{it}(\tstar)\le1-S_{it}$ and $\psi_{i,t}(\mu_t(\tstar))\le S_{it}$

\begin{equation*}
    \frac{\delta_{it}(\tstar)}{S_{it}} \le\frac{1-S_{it}}{S_{it}} \le\frac{1-S_{it}}{\psi_{i,t}(\mu_t(\tstar))} =(1-S_{it})\,\phi_{i,t}(\nu_t) =\Delta_{2,it}.
\end{equation*}

Multiplying by $v_{i,t+2}>0$ and summing, $\gamma_t\le\sum_iv_{i,t+2}\Delta_{2,it}\le D_t$. Finally, by lemma \ref{lem:sumD} and Tonelli's Theorem, $\Estar[\sum_t\gamma_t]=\sum_t\Estar[\gamma_t]\le\sum_t\Estar[D_t]\le Y_0<\infty$, so $\sum_t\gamma_t<\infty$ $\mathbb{P}^{\tstar}$-almost surely. We can then use the Robbins--Siegmund Theorem \citep{robbins1971convergence} and obtain

\begin{equation*}
    \sum_{t=0}^{\infty}\sum_{i=1}^{n}v_{i,t+2}\zeta_{it}(\theta^*)\chi_i^2(t)<+\infty,
\end{equation*}

together with the convergence of $V_t$ to some random variable $V_{\infty}$; hence $X_t$ converges almost surely to $X_{\infty}:=1-V_{\infty}$.

Denote by $K_t$ the term on which we focus:

\begin{equation*}
    K_t:=\sum_{i=1}^{n}v_{i,t+2}\zeta_{it}(\theta^*)\chi_i^2(t).
\end{equation*}

One has

\begin{equation*}
    0\le K_t\le V_t-\mathbb E_t(V_{t+1})+\gamma_t,
\end{equation*}

so that

\begin{equation*}
    0\le\mathbb E^{\theta^*}(K_t)\le\mathbb E^{\theta^*}(V_t)-\mathbb E^{\theta^*}(V_{t+1})+\mathbb E^{\theta^*}(\gamma_t).
\end{equation*}

Therefore, for $N\in\mathbb N$,

\begin{equation*}
    0\le\sum_{t=0}^{N-1}\mathbb E^{\theta^*}(K_t) \le\mathbb E^{\theta^*}(V_0)-\mathbb E^{\theta^*}(V_N)+\sum_{t=0}^{N-1}\mathbb E^{\theta^*}(\gamma_t) \le V_0+\sum_{t=0}^{\infty}\mathbb E^{\theta^*}(\gamma_t)<+\infty,
\end{equation*}

which gives the result.
\end{proof}

\begin{lemma}[Weight transport]\label{lem:transport}
Let

\begin{equation*}
    \kappa:=\bigl((1-\abar)c_\ell\bigr)^{-1}.
\end{equation*}

Then, for all $i$ and $t$,

\begin{equation}\label{eq:transportpath}
    \zeta_{it}(\tstar)
    \le
    \kappa\Bigl(\zeta_{i,t+1}(\tstar)+\bigl|\epsilon_{i,t+1}(\tstar)\bigr|\Bigr).
\end{equation}
\end{lemma}

\begin{proof}
By the definition of $\epsilon_{i,t+1}$ in \eqref{eq:deltaeps},

\begin{equation*}
    \zeta_{i,t+1}(\tstar) = \sum_j a_{ij,t+2}\mu_{j,t+1}(\tstar) +\epsilon_{i,t+1}(\tstar).
\end{equation*}

Since the coefficients $a_{ij,t+2}$ are non-negative,

\begin{equation*}
    \zeta_{i,t+1}(\tstar) \ge a_{ii,t+2}\mu_{i,t+1}(\tstar) - \bigl|\epsilon_{i,t+1}(\tstar)\bigr|.
\end{equation*}

Using the Bayesian update \eqref{eq:bayes},

\begin{equation*}
    \mu_{i,t+1}(\tstar) = \zeta_{it}(\tstar) \frac{\ell_i^{\tstar}(\omega_{i,t+1})} {m_{it}(\omega_{i,t+1})},
\end{equation*}

combined with the bounds

\begin{equation*}
    m_{it}(\omega)\le1,\qquad \ell_i^{\tstar}(\omega)\ge c_\ell,\qquad a_{ii,t+2}\ge1-\abar,
\end{equation*}

the last one following from Hypothesis~\ref{hyp:H2}, we obtain

\begin{equation*}
    a_{ii,t+2}\mu_{i,t+1}(\tstar) \ge (1-\abar)c_\ell\,\zeta_{it}(\tstar) = \kappa^{-1}\zeta_{it}(\tstar).
\end{equation*}

Therefore,

\begin{equation*}
    \zeta_{i,t+1}(\tstar)
    \ge
    \kappa^{-1}\zeta_{it}(\tstar)
    -
    \bigl|\epsilon_{i,t+1}(\tstar)\bigr|,
\end{equation*}

which is equivalent to \eqref{eq:transportpath}.
\end{proof}

We then define the following terms:

\begin{equation}\label{eq:Zdef}
    Z^{(k)}_{it}(\omega^k):=\sum_{\that\in\Theta}\zeta_{it}(\that)\,\Pi^{(k)}_i(\that)-1,
    \qquad
    \Pi^{(k)}_i(\that):=\prod_{r=1}^{k}\frac{\ell^{\that}_i(\omega_r)}{\ell^{\tstar}_i(\omega_r)} .
\end{equation}

\begin{proposition}\label{prop:Zsum}
\footnote{Our proof is very similar to the proof of Lemma A.5 of \citet{molavi2018theory}. We have to deal with an extra error term and we use the ordinary expectation rather than the conditional one. Moreover, we believe that there is a mistake in their base case, because they remove the term $\zeta_{i,t}(\theta^*)$ from the conditional expectation (their equations (39) and (40)), arguing that this term is bounded from below, which is the case, but it is bounded from below by a random variable and not by a deterministic constant.} For every $k\ge1$, all $i\in\mathcal N$, and for any $(\omega_1,\dots,\omega_k)\in\Omega_i^k$,

\begin{equation*}
    \sum_{t=0}^{\infty}\mathbb E^{\theta^*}\!\left[\Bigl(\sum_{\theta\in\Theta}\zeta_{it}(\theta)\prod_{r=1}^{k}\frac{\ell_i^{\theta}(\omega_r)}{\ell_i^{\theta^*}(\omega_r)}-1\Bigr)^{2}\zeta_{it}(\theta^*)\right]<+\infty.
\end{equation*}
\end{proposition}

\begin{proof}
We proceed by induction on $k$.

\noindent$\bullet$\ \textbf{Base case ($k=1$).}
One has

\begin{equation*}
    Z_{it}^{(1)}(\omega_1)=\frac{m_{it}(\omega_1)-\ell_i^{\theta^*}(\omega_1)}{\ell_i^{\theta^*}(\omega_1)},
\end{equation*}

so that, by \eqref{eq:chi2},

\begin{equation*}
    \bigl(Z_{it}^{(1)}(\omega_1)\bigr)^2\le\frac{\chi_{i}^2(t)}{c_\ell^2},
\end{equation*}

and therefore

\begin{equation*}
    \zeta_{it}(\theta^*)\bigl(Z_{it}^{(1)}(\omega_1)\bigr)^2\le\frac{1}{c_\ell^2}\zeta_{it}(\theta^*)\chi_i^2(t).
\end{equation*}

We then use Proposition~\ref{prop:chi} to conclude for the base case.

Introduce the following $\sigma$-algebra containing everything except agent $i$'s fresh signal,

\begin{equation*}
    \cF_{-i,t+1}:=\sigma\Bigl(\cF_t\cup\sigma\bigl(\omega_{j,t+1}:j\neq i\bigr)\Bigr),
    \qquad
    \Emi[\,\cdot\,]:=\Estar[\,\cdot\mid\cF_{-i,t+1}] .
\end{equation*}

\noindent$\bullet$\ \textbf{Induction step.} Suppose the result holds for $k\ge1$. Take $(\omega_1,\dots,\omega_{k+1})\in\Omega_i^{k+1}$. One has

\begin{equation*}
    \Emi[\zeta_{i,t+1}(\theta)\bigr]
    =\sum_{j=1}^{n}a_{ij,t+2}\zeta_{jt}(\theta)\,\Emi\left[\frac{\ell_j^{\theta}(\omega_{j,t+1})}{m_{jt}(\omega_{j,t+1})}\right]
    +\Emi\bigl[\epsilon_{i,t+1}\bigr],
\end{equation*}

so that, using the Hypothesis of independence between agents and the $\cF_{-i,t+1}$-measurability of the terms where $j\neq i$,

\begin{equation*}
    a_{ii,t+2}\zeta_{it}(\theta)\!\left(\frac{\ell_i^{\theta}(\omega_{i,t+1})}{m_{it}(\omega_{i,t+1})}
    -\Emi\!\left[\frac{\ell_i^{\theta}(\omega_{i,t+1})}{m_{it}(\omega_{i,t+1})}\right]\right)
\end{equation*}

\begin{equation*}
    =\Bigl(\zeta_{i,t+1}(\theta)-\Emi\bigl[\zeta_{i,t+1}(\theta)\bigr]\Bigr)
    -\Bigl(\epsilon_{i,t+1}(\theta)-\Emi\bigl[\epsilon_{i,t+1}(\theta)\bigr]\Bigr).
\end{equation*}

Define

\begin{equation*}
    R^{\epsilon}_{t+1}:=\sum_{\that\in\Theta}\epsilon_{i,t+1}(\that)\,\Pi^{(k)}_i(\that),
    \qquad
    \Upsilon_{it}:=\sum_{\that}\zeta_{it}(\that)\,\Pi^{(k)}_i(\that)\,\frac{\ell^{\that}_i(\omega_{i,t+1})}{\ell^{\tstar}_i(\omega_{i,t+1})}-1 ,
\end{equation*}

and

\begin{align}
    T_1&:=a_{ii,t+2}\sum_{\that\in\Theta}\zeta_{it}(\that)\,\Pi^{(k)}_i(\that) \left(\frac{\ell^{\that}_i(\omega_{i,t+1})}{\ell^{\tstar}_i(\omega_{i,t+1})} -\frac{\ell^{\that}_i(\omega_{i,t+1})}{m_{it}(\omega_{i,t+1})}\right), \label{eq:T1def}\\
    T_2&:=\Emi[T_1], \\
    T_3&:=a_{ii,t+2}\,Z^{(k)}_{it}(\omega^k), \label{eq:T23def}\\
    T_4&:=Z^{(k)}_{i,t+1}(\omega^k)-\Emi\bigl[Z^{(k)}_{i,t+1}(\omega^k)\bigr], \\
    T_5&:=R^{\epsilon}_{t+1}-\Emi\bigl[R^{\epsilon}_{t+1}\bigr]. \label{eq:T45def}
\end{align}

Then one has

\begin{equation*}
    a_{ii,t+2}\,\Upsilon_{it}=T_1-T_2+T_3+T_4-T_5 .
\end{equation*}

We now square this equality, use the inequality $(x_1+x_2+x_3+x_4+x_5)^2\le5(x_1^2+x_2^2+x_3^2+x_4^2+x_5^2)$, multiply the result by $\zeta_{i,t}(\theta^*)$ and take the expectation to obtain

\begin{equation}\label{eq:synth}
    a_{ii,t+2}^2\,\Estar\bigl[\zeta_{it}(\tstar)\,\Upsilon_{it}^2\bigr]
    \;\le\;5\sum_{m=1}^{5}\Estar\bigl[\zeta_{it}(\tstar)\,T_m^2\bigr].
\end{equation}

We show that the five terms on the right-hand side are summable.

\begin{itemize}
    \item \emph{Term $T_1$.} One has
    
    \begin{equation*}
        \frac{\ell_i^{\hat{\theta}}(\omega_{i,t+1})} {\ell_i^{\tstar}(\omega_{i,t+1})} - \frac{\ell_i^{\hat{\theta}}(\omega_{i,t+1})} {m_{it}(\omega_{i,t+1})} = \frac{\ell_i^{\hat{\theta}}(\omega_{i,t+1})} {m_{it}(\omega_{i,t+1})} \left( \sum_{\tilde{\theta}\in\Theta} \zeta_{it}(\tilde{\theta}) \frac{\ell_i^{\tilde{\theta}}(\omega_{i,t+1})} {\ell_i^{\tstar}(\omega_{i,t+1})} -1 \right).
    \end{equation*}
    
    Using the tower property and the $\mathcal{F}_t$-measurability of $\zeta_{i,t}(\theta^*)$,
    
    \begin{equation*}
        \Estar\bigl[\zeta_{it}(\tstar)\,T_1^2\bigr]=\Estar\bigl[\zeta_{it}(\tstar)\,\Emi\bigl[ T_1^2\bigr]\bigr]
    \end{equation*}
    
    \begin{equation*}
        = \Estar\Biggl[ \zeta_{it}(\tstar)\, \sum_{\omega\in\Omega_i}\Biggl[l_i^{\theta^*}(\omega) \Biggl( a_{ii,t+2} \sum_{\that\in\Theta} \zeta_{it}(\that)\, \Pi_i^{(k)}(\that) \frac{\ell_i^{\hat{\theta}}(\omega)} {m_{it}(\omega)}\left( Z_{i,t}^{(1)}(\omega) \right) \Biggr)^2 \Biggr] \Biggr].
    \end{equation*}
    
    Since $l_i^{\theta}$ is bounded above by $1$ for all $\theta\in\Theta$, $a_{ii,t+2}$ is bounded above by $1$, $\zeta_{i,t}(\theta)$ is bounded above by $1$ for all $\theta\in\Theta$, and $\Pi_i^{(k)}(\that)$ is also bounded and $m_{i,t}$ is also bounded from below by $c_l$. Therefore, there exists a constant $\tilde{c}>0$ such that
    
    \begin{equation*}
        \Estar\bigl[\zeta_{it}(\tstar)\,T_1^2\bigr]\le \tilde{c}\sum_{\omega\in\Omega_i}\Estar\bigl[\zeta_{it}(\tstar)\,Z_{i,t}^{(1)}(\omega)^2\bigr].
    \end{equation*}
    
    We can then conclude using the base case.
    
    \item \emph{Term $T_2$.} By conditional Jensen, $T_2^2=(\Emi[T_1])^2\le\Emi[T_1^2]$. Since $\zeta_{it}(\tstar)$ is $\cF_{-i,t+1}$-measurable, the tower property gives
    
    \begin{equation*}
        \Estar\bigl[\zeta_{it}(\tstar)T_2^2\bigr]
        \le\Estar\bigl[\zeta_{it}(\tstar)\,\Emi[T_1^2]\bigr]
        =\Estar\bigl[\Emi[\zeta_{it}(\tstar)T_1^2]\bigr]
        =\Estar\bigl[\zeta_{it}(\tstar)T_1^2\bigr],
    \end{equation*}
    
    and we conclude by the case of $T_1$.
    
    \item \emph{Term $T_3$.} Since $a_{ii,t+2}$ is bounded above by $1$, we conclude immediately using the induction Hypothesis.
    
    \item \emph{Term $T_4$.} The conditional variance is dominated by the conditional second moment: $\Emi[T_4^2]\le\Emi[(Z^{(k)}_{i,t+1}(\omega^k))^2]$. By the tower property, as for $T_2$,
    
    \begin{equation*}
        \Estar\bigl[\zeta_{it}(\tstar)T_4^2\bigr]
        \le\Estar\Bigl[\zeta_{it}(\tstar)\bigl(Z^{(k)}_{i,t+1}(\omega^k)\bigr)^2\Bigr].
    \end{equation*}
    
    Moreover, by Lemma~\ref{lem:transport}, $\zeta_{it}(\tstar)\le\kappa\bigl(\zeta_{i,t+1}(\tstar)+|\epsilon_{i,t+1}(\tstar)|\bigr)$. Then
    
    \begin{equation*}
        \Estar\bigl[\zeta_{it}(\tstar)T_4^2\bigr]
        \le\kappa\Estar\Bigl[\zeta_{i,t+1}(\tstar)\bigl(Z^{(k)}_{i,t+1}(\omega^k)\bigr)^2\Bigr]+\kappa\Estar\Bigl[\bigl|\epsilon_{i,t+1}(\tstar)\bigr|\bigl(Z^{(k)}_{i,t+1}(\omega^k)\bigr)^2\Bigr].
    \end{equation*}
    
    The first term on the right-hand side is summable by the induction Hypothesis, and the second term is summable because $\bigl|\epsilon_{i,t+1}(\tstar)\bigr|$ is summable by Lemma~\ref{lem:remainders}, and $\bigl(Z^{(k)}_{i,t+1}(\omega^k)\bigr)^2$ is bounded uniformly in time.
    
    \item \emph{Term $T_5$.} As for $T_4$, $\Emi[T_5^2]\le\Emi[(R^{\epsilon}_{t+1})^2]$. The tower property then gives
    
    \begin{equation*}
        \Estar\bigl[\zeta_{it}(\tstar)T_5^2\bigr]
        =\Estar\bigl[\zeta_{it}(\tstar)\,\Emi[T_5^2]\bigr]
        \le\Estar\bigl[\Emi[(R^{\epsilon}_{t+1})^2]\bigr]
        =\Estar\bigl[(R^{\epsilon}_{t+1})^2\bigr].
    \end{equation*}
    
    The factor $\Pi_i^{(k)}(\that)$ is bounded, and since $\Estar\bigl[\bigl|\epsilon_{i,t+1}(\tstar)\bigr|\bigr]$ is summable by Lemma~\ref{lem:remainders}, the same holds for $\Estar\bigl[\epsilon_{i,t+1}(\tstar)^2\bigr]$, because $\epsilon_{i,t+1}(\theta)^2\le\bigl|\epsilon_{i,t+1}(\theta)\bigr|$; for all state $\theta$ we can then conclude.
\end{itemize}

Since $a_{ii,t+1}\ge(1-\bar{a})$ by Hypothesis~\ref{hyp:H2}, \eqref{eq:synth} gives $\sum_{t=1}^\infty\Estar\bigl[\zeta_{it}(\tstar)\,\Upsilon_{it}^2\bigr]<\infty$. Using the tower property and the $\mathcal{F}_t$-measurability of $\zeta_{i,t}(\theta^*)$, we end up with

\begin{equation*}
    \sum_{t=1}^\infty\Estar\bigl[\zeta_{it}(\tstar)\,\E_t[\Upsilon_{it}^2]\bigr]=\sum_{t=1}^\infty\sum_{\omega\in\Omega_i}l_i^{\tstar}(\omega)\,\Estar\bigl[\zeta_{it}(\tstar)Z_{i,t}^{(k+1)}(\omega^k,\omega)^2\bigr]<\infty.
\end{equation*}

By Hypothesis~\ref{hyp:H3}, we obtain that for all $\omega\in\Omega_i$,

\begin{equation*}
    \sum_{t=1}^\infty\Estar\bigl[\zeta_{it}(\tstar)Z_{i,t}^{(k+1)}(\omega^k,\omega)^2\bigr]<\infty,
\end{equation*}

which concludes the induction step.
\end{proof}

We are now able to prove Proposition~\ref{prop:vanish}.

\begin{proof}[Proof of Proposition~\ref{prop:vanish}]
We use Lemma 4 of \citet{jadbabaie2012non}: there exist an integer $k$, a tuple $(\omega_1,\omega_2,\dots,\omega_k)\in\Omega_i^k$ and a $\beta\in(0,1)$ such that $\prod_{r=1}^k\frac{l_i^{\theta}(\omega_r)}{l_i^{\theta^*}(\omega_r)}-1<-\beta$ for all $\theta\in\Theta_i$. Therefore, using Proposition~\ref{prop:Zsum} with $(\omega_1,\omega_2,\dots,\omega_k)$, we obtain

\begin{equation*}
    \beta^2\Estar\Bigl[\zeta_{it}(\tstar)\Bigl(\sum_{\that\in\Theta_i}\zeta_{it}(\that)\Bigr)^2\Bigr]\le\Estar\Bigl[\zeta_{it}(\tstar)\Bigl(\sum_{\that\in\Theta}\zeta_{it}(\that)\,(\Pi^{(k)}_i(\that)-1)\Bigr)^2\Bigr].
\end{equation*}

Since $\beta>0$, for all $\that\in\Theta_i$ we have

\begin{equation*}
    \sum_{t=1}^{\infty}\Estar\bigl[\zeta_{it}(\tstar)\zeta_{it}(\that)^2\bigr]<\infty .
\end{equation*}

As the terms inside the expectations are non-negative, this is equivalent, by Tonelli's Theorem, to

\begin{equation*}
    \Estar\Bigl[\sum_{t=1}^{\infty}\zeta_{it}(\tstar)\zeta_{it}(\that)^2\Bigr]<\infty,
\end{equation*}

which in turn implies

\begin{equation*}
    \sum_{t=1}^{\infty}\zeta_{it}(\tstar)\zeta_{it}(\that)^2<\infty
\end{equation*}

$\mathbb{P}^{\tstar}$-almost surely. By Proposition~\ref{prop:smg}, we know that $\zeta_{i,t}(\tstar)$ is bounded from below by $C>0$, because 

\begin{equation*}
    \zeta_{it}(\tstar)=\frac{\psi_{i,t}\bigl(\mu_t(\tstar)\bigr)}{S_{it}}
    \;\ge\;\psi_{i,t}\bigl(\mu_t(\tstar)\bigr)
    \;\ge\;\psi_{i,t}\bigl(C\mathbf{1}\bigr)
    \;=\;C\,\psi_{i,t}(\mathbf{1})\;=\;C\;>\;0
    \qquad \Pstar\text{-a.s.}
\end{equation*}

$\mathbb{P}^{\tstar}$-almost surely. Then

\begin{equation*}
    \sum_{t=1}^{\infty}\zeta_{it}(\that)^2<\infty,
\end{equation*}

which concludes the proof of Proposition~\ref{prop:vanish}.
\end{proof}

The second part of the proof consists in using this fact, together with the strong connectivity of the communication network, to make this result propagate. To this end we shall need some lemmas. First, we recall the definition of the elasticity matrices given in Definition~\ref{def:curv} (taken at the point $\mu_t(\tstar)$. For agents $i,k$ and time $t$, set

\begin{equation*}
    e_{ik,t}:=e_{ik,t}\bigl(\mu_t(\tstar)\bigr)
    =\frac{\partial\log\psi_{i,t}}{\partial\log x_k}\bigl(\mu_t(\tstar)\bigr)=\frac{\mu_{kt}(\tstar)\,\partial_{x_k}\psi_{i,t}\bigl(\mu_t(\tstar)\bigr)}{\psi_{i,t}\bigl(\mu_t(\tstar)\bigr)}\;\ge\;0,
    \qquad
    E_t:=[e_{ik,t}]_{i,k\in\cN},
\end{equation*}

and recall

\begin{equation*}
    H^{i,t}_{jk}(x):=\frac{\partial^2\log\psi_{i,t}}{\partial\log x_j\,\partial\log x_k}(x).
\end{equation*}

These objects satisfy the following properties.

\begin{lemma}\label{lem:elast}
For every $x>0$, $\sum_{k}e_{ik,t}(x)=1$, in particular $E_t$ is row-stochastic, $e_{ik,t}(\bone)=a_{ik,t+1}$, and $\sum_j H^{i,t}_{jk}(x)=0$ for every $k$ and every $x>0$.
\end{lemma}

\begin{proof}
Euler's identity for the degree-one homogeneous function $\psi_{i,t}$ reads $\sum_kx_k\partial_{x_k}\psi_{i,t}(x)=\psi_{i,t}(x)$; dividing by $\psi_{i,t}(x)>0$ we obtain that the sum equals $1$. Differentiating this expression with respect to $\log(x_k)$ immediately gives $\sum_j H^{i,t}_{jk}(x)=0$. Evaluating the elasticity at $x=\bone$ and using $\psi_{i,t}(\bone)=1$ from Hypothesis~\ref{hyp:H1} gives $e_{ik,t}(\bone)=\partial_{x_k}\psi_{i,t}(\bone)=a_{ik,t+1}$ by Definition~\ref{def:weights}.
\end{proof}

\begin{lemma}\label{lem:comparison}
Almost surely, for all $i,k,t$,

\begin{equation}\label{eq:comparison}
    C^2\,a_{ik,t+1}\;\le\;e_{ik,t}\;\le\;C^{-2}\,a_{ik,t+1},
\end{equation}

where $C$ is the random constant of Proposition~\textup{\ref{prop:smg}}.
\end{lemma}

\begin{proof}
Fix $i,k,t$ and abbreviate $e(x):=e_{ik,t}(x)$. Wherever $e(x)>0$, $\log e$ is continuously differentiable with

\begin{equation*}
    \frac{\partial\log e}{\partial\log x_j}(x)=\frac{H^{i,t}_{jk}(x)}{e(x)} .
\end{equation*}

For $j\neq k$ the curvature bound of Hypothesis~\ref{hyp:H1} gives $|H^{i,t}_{jk}|/e\le e_{ij,t}(x)$; for $j=k$, Lemma~\ref{lem:elast} gives $H^{i,t}_{kk}=-\sum_{j\neq k}H^{i,t}_{jk}$, so $|H^{i,t}_{kk}|\le e(x)\sum_{j\neq k}e_{ij,t}(x)=e(x)(1-e(x))$, again by Lemma~\ref{lem:elast}. Summing,

\begin{equation}\label{eq:gradbound}
    \sum_{j=1}^n\left|\frac{\partial\log e}{\partial\log x_j}(x)\right|
    \;\le\;\sum_{j\neq k}e_{ij,t}(x)+\bigl(1-e(x)\bigr)
    =2\bigl(1-e(x)\bigr)\;\le\;2 .
\end{equation}

Work on the almost-sure event $\{\mu_t(\tstar)\in[C,1]^n\}$ given by Proposition~\ref{prop:smg}, and join $\bone$ to $\mu_t(\tstar)$ by the log-segment

\begin{equation*}
    x(s):=\bigl(\mu_{jt}(\tstar)^s\bigr)_{j=1}^n,\qquad s\in[0,1],
\end{equation*}

so that $x(0)=\bone$, $x(1)=\mu_t(\tstar)$, and $\frac{d}{ds}\log x_j(s)=\log\mu_{jt}(\tstar)$ with $|\log\mu_{jt}(\tstar)|\le\log(1/C)$. Wherever $e(x(s))>0$, the function $f(s):=\log e(x(s))$ satisfies, by \eqref{eq:gradbound},

\begin{equation}\label{eq:fprime}
    |f'(s)|=\Bigl|\sum_j\frac{\partial\log e}{\partial\log x_j}\bigl(x(s)\bigr)\log\mu_{jt}(\tstar)\Bigr|
    \;\le\;2\log(1/C) .
\end{equation}

We now distinguish two cases.

\begin{itemize}
    \item $a_{ik,t+1}>0$ (meaning $k\in\mathcal{N}_i$). In this case $\log(a_{i,k,t+1})$ is well defined and $e(x(s))>0$ by strict monotonicity. The mean value Theorem gives the existence of $s\in(0,1)$ such that, by \eqref{eq:fprime},
    
    \begin{equation*}
        |f(1)-f(0)| \le (1-0)\times|f'(s)| \le2\log(1/C).
    \end{equation*}
    
    And exponentiating allows us to conclude in this case.
    
    \item $a_{ik,t+1}=0 $ (meaning $k \notin \mathcal{N}_i$), in this case, $e_{ik,t}=0$ since $\psi_{i,t}$ does not depend on $x_k$ by Hypothesis \ref{hyp:H1} and the inequalities simply read $0\le0\le0$ which is true.
\end{itemize}
\end{proof}

\begin{corollary}[Support equivalence]\label{cor:support}
Almost surely, for all $i,k,t$: $a_{ik,t+1}=0\iff e_{ik,t}=0$ by Hypothesis \ref{hyp:H2}. In particular, at every $t$ the matrix $E_t$ has the same zero pattern as the weights matrix $A_{t+1}$.
\end{corollary}

\begin{proof}
By Lemma~\ref{lem:comparison}: if $a_{ik,t+1}=0$ then $0\le e_{ik,t}\le C^{-2}\cdot0=0$; if $e_{ik,t}=0$ then $C^2a_{ik,t+1}\le0$, and since $C>0$ and $a_{ik,t+1}\ge0$ this forces $a_{ik,t+1}=0$.
\end{proof}

Let us now fix $\that\in\Theta$ different from the true state $\theta^*$, and define

\begin{equation*}
    r_{it}:=\frac{\mu_{it}(\that)}{\mu_{it}(\tstar)}\;\ge\;0,
    \qquad
    D:=\{i\in\cN:\that\in\Theta_i\},
    \qquad
    D^c:=\cN\setminus D .
\end{equation*}

By the identifiability part of Hypothesis~\ref{hyp:H3}, $D\neq\emptyset$. Let

\begin{equation*}
    \beta_t:=\max_{i\in D}r_{it},
    \qquad
    u_t:=(r_{it})_{i\in D^c} .
\end{equation*}

Using Hypothesis~\ref{hyp:H2} together with Lemma~\ref{lem:comparison}, we define the following quantities:

\begin{equation}\label{eq:elastbounds}
    \begin{gathered}
        e_{ik,t}\ \ge\ \alpha_t:=C^2\,\aund\,\lambda_t \quad\text{for }(i,k)\in\cE,\ i\neq k,
        \qquad
        \sum_{k\neq i}e_{ik,t}\ \le\ K\lambda_t,\quad K:=C^{-2}\,\abar,\\
        e_{ii,t}\ \ge\ h_1:=C^2(1-\abar)>0,
        \qquad
        e_{ii,t}\ \ge\ 1-K\lambda_t .
    \end{gathered}
\end{equation} 

Note that since $t\lambda_t\to\infty$, eventually $\lambda_t\ge1/t$, so $\sum_t\lambda_t=\infty$ and hence $\sum_t\alpha_t=\infty$ $\mathbb{P}^{\theta^*}$-almost surely; also $0<\alpha_t\le\lambda_t\le1$.

\begin{lemma}\label{lem:recursion}
One has

\begin{equation}\label{eq:recursion}
    u_{t+1}\ \le\ P_t\,u_t+\beta_t\,g_t,
    \qquad
    P_t:=[e_{ik,t}]_{i,k\in D^c}\ \ge\ 0,
    \qquad
    g_{it}:=\sum_{k\in D}e_{ik,t}\ \ge\ 0 ,
\end{equation} 

and $P_t\bone=\bone-g_t$.
\end{lemma}

\begin{proof}
We divide the Bayes update \eqref{eq:bayes} for $\that$ by the Bayes update for $\tstar$ to obtain

\begin{equation}\label{eq:ratioupdate}
    r_{i,t+1}
    =L_{i,t+1}\cdot\frac{\zeta_{it}(\that)}{\zeta_{it}(\tstar)},
    \qquad
    L_{i,t+1}:=\frac{\ell^{\that}_i(\omega_{i,t+1})}{\ell^{\tstar}_i(\omega_{i,t+1})} .
\end{equation}

For $i\in D^c$ one has $L_{i,t+1}=1$ by definition of $D^c$ and Hypothesis~\ref{hyp:H3}, therefore

\begin{equation*}
    r_{i,t+1}=\frac{\zeta_{it}(\that)}{\zeta_{it}(\tstar)}
    =\frac{\psi_{i,t}\bigl(\mu_t(\that)\bigr)}{\psi_{i,t}\bigl(\mu_t(\tstar)\bigr)} .
\end{equation*}

Using the concavity and the degree-one homogeneity of $\psi_{i,t}$, together with Definition~\ref{def:curv}, we obtain

\begin{equation*}
    r_{i,t+1}
    \;\le\;\sum_{k=1}^n\frac{\mu_{kt}(\tstar)\,\partial_{x_k}\psi_{i,t}\bigl(\mu_t(\tstar)\bigr)}{\psi_{i,t}\bigl(\mu_t(\tstar)\bigr)}\cdot\frac{\mu_{kt}(\that)}{\mu_{kt}(\tstar)}
    =\sum_{k=1}^n e_{ik,t}\,r_{kt} .
\end{equation*}

We separate the sum and bound $r_{kt}$ from above by $\beta_t$ for $k\in D$, we obtain the first assertion. Since $E_t$ is row-stochastic by Lemma~\ref{lem:elast}, we immediately get $P_t\bone=\bone-g_t$.
\end{proof}

\begin{lemma}\label{lem:telescope}
Let $\Phi(t,\tau):=P_{t-1}P_{t-2}\cdots P_{\tau}$ (with $\Phi(t,t):=I$). For $\tau\le t$,

\begin{equation*}
    \sum_{s=\tau}^{t-1}\Phi(t,s+1)\,\beta_s\,g_s\;\le\;\Bigl(\sup_{s\ge\tau}\beta_s\Bigr)\,\bone .
\end{equation*}
\end{lemma}

\begin{proof}
Because $g_s=\bone-P_s\bone$ by Lemma~\ref{lem:recursion}, and by the regrouping $\Phi(t,s+1)P_s=\Phi(t,s)$,

\begin{equation*}
    \Phi(t,s+1)\,g_s=\Phi(t,s+1)\bone-\Phi(t,s)\bone=\varphi_{s+1}-\varphi_s,
    \qquad
    \varphi_s:=\Phi(t,s)\bone .
\end{equation*}

Moreover, each $\Phi(t,s+1)$ has positive entries and $g_s\ge0$, so $\varphi_{s+1}-\varphi_s$ is componentwise non-negative and the sum telescopes:

\begin{equation*}
    \sum_{s=\tau}^{t-1}\bigl(\varphi_{s+1}-\varphi_s\bigr)=\varphi_t-\varphi_\tau=\bone-\Phi(t,\tau)\bone\;\le\;\bone .
\end{equation*}

Therefore $\sum_s\beta_s(\varphi_{s+1}-\varphi_s)\le(\sup_{s\ge\tau}\beta_s)\sum_s(\varphi_{s+1}-\varphi_s)\le(\sup_{s\ge\tau}\beta_s)\bone$.
\end{proof}

\begin{lemma}\label{lem:contraction}
Assume $D^c\neq\emptyset$. Then, almost surely, for every fixed $\tau$,

\begin{equation*}
    \norminf{\Phi(t,\tau)}\;\xrightarrow[t\to\infty]{}\;0 ,
\end{equation*}

where $\norminf{\cdot}$ denotes the maximum row sum.
\end{lemma}

\begin{proof}
All matrices have positive terms, also, $\norminf{\cdot}$ is submultiplicative and one has $\norminf{P_s}\le1$ since $\sum_{k\in D^c}e_{ik,s}\le\sum_{k\in\cN}e_{ik,s}=1$ by Lemma~\ref{lem:elast}. Also,

\begin{equation*}
    \Phi(t,\tau)_{ij}=\sum_{(m_s)}\ \prod_{s=\tau}^{t-1}e_{m_{s+1}m_s,\,s},
    \qquad m_t=i,\ m_\tau=j,
\end{equation*}

the sum is defined over paths $m_s$ having the following form: they start at $j$, finish at $i$, and remain in $D^c$ at each step, but in the reversed graph where all the directions are swapped (because in the matrix product defining $\Phi(t,\tau)$ the matrices with latest times are in the left of the product), or, equivalently, we say that it starts at $i$ and finishes at $j$ but running backward in time. Since $\sum_s\alpha_s=\infty$ $\mathbb{P}^{\tstar}$-almost surely and $0<\alpha_s<1$, one can partition $[\tau,\infty)$ into consecutive blocks $J_0,J_1,\dots$ with

\begin{equation*}
    \sum_{s\in J_q}\alpha_s\;\in\;\bigl[\,2(n-1),\,2(n-1)+1\,\bigr]
\end{equation*}

and there are infinitely many such blocks. Now, let us show that $\bigl\lVert\prod_{s\in J}P_s\bigr\rVert_\infty\le1-p_\star$ for each of these blocks, with $p_\star>0$ defined below, independent of the block. Fix a block $J$. Since each $\alpha_s<1$ and $\sum_{J}\alpha_s\ge2(n-1)$, a segment of $J$ can be partitioned into $n-1$ consecutive sub-blocks $I_1,\dots,I_{n-1}$ verifying

\begin{equation*}
    \sum_{s\in I_\ell}\alpha_s\in[1,2)\qquad(\ell=1,\dots,n-1)
\end{equation*}

(at most, after $m\le n-2$ sub-blocks $2m$ of the budget has been used, so at least $2$ remains in the total sum); let $W:=I_1\cup\dots\cup I_{n-1}$, with $I_1$ containing the latest times then $I_2$ the latest times of what remains, up until $I_{n-1}$ containing the earliest times. Note that $W$ depends only on the block $J$, and not on any row index. By submultiplicativity and $\norminf{P_s}\le1$, $\bigl\lVert\prod_{s\in J}P_s\bigr\rVert_\infty\le\bigl\lVert\prod_{s\in W}P_s\bigr\rVert_\infty$, so it is sufficient to bound the norm of $\prod_{s\in W}P_s$. Now fix a row index $i\in D^c$. By the strong connectivity of the graph and the fact that $D\neq\emptyset$, choose a shortest directed path $i=p_0\to p_1\to\cdots\to p_L=q\in D$, so that $L=\mathrm{dist}_{\cG}(i,D)\le n-1$, $p_0,\dots,p_{L-1}\in D^c$, and $e_{p_{\ell-1}p_\ell,s}\ge\alpha_s$ by \eqref{eq:elastbounds}.

Since $\prod_{s\in W}E_s$ is row-stochastic by Lemma~\ref{lem:elast}, the total weight of all walks starting from $i$ (in the latest time and moving backward in time) over the horizon $W$ equals one. Splitting the walks according to whether or not they stay in $D^c$ or visit $D$ at least once,

\begin{equation}\label{eq:Qsplit}
    \sum_{j\in D^c}\Bigl(\prod_{s\in W}P_s\Bigr)_{ij}=1-Q_i,
    \qquad
    Q_i:=\ \text{total weight of walks from $i$ that visit $D$}\ \ge0 .
\end{equation}

Restrict $Q_i$ to the following walks: within each sub-block $I_\ell$ with $\ell\le L$, the walk stays at $p_{\ell-1}$ until a chosen period $s_\ell\in I_\ell$, steps to $p_\ell$ at $s_\ell$, and stays at $p_\ell$ for the remainder times of $I_\ell$; after reaching $q=p_L$ it stays at $q$ during the whole of $I_{L+1},\dots,I_{n-1}$. Distinct choices of jump times $(s_1,\dots,s_L)$ give distinct walks, each of which visits $D$; since the position at the boundaries of the sub-blocks is pinned, the total weight factorizes over the sub-blocks:

\begin{equation}\label{eq:Sfactor}
    Q_i\;\ge\;\prod_{\ell=1}^{L}S_\ell\;\prod_{m=L+1}^{n-1}\;\prod_{s\in I_m}e_{qq,s},
    \qquad
    S_\ell:=\sum_{s\in I_\ell}
    \Bigl(\prod_{\substack{s'\in I_\ell\\ s'>s}}e_{p_{\ell-1}p_{\ell-1},s'}\Bigr)
    e_{p_{\ell-1}p_\ell,s}
    \Bigl(\prod_{\substack{s'\in I_\ell\\ s'<s}}e_{p_\ell p_\ell,s'}\Bigr).
\end{equation}

Put $\sigma_s:=\max\{h_1,\,1-K\lambda_s\}\le1$. By \eqref{eq:elastbounds}, $e_{jj,s}\ge\sigma_s$ (for all $j\in\mathcal{N}$, in particular for $q\in D$). Therefore, since $e_{p_{\ell-1}p_\ell,s}\ge\alpha_s$ and $\sum_{s\in I_\ell}\alpha_s\ge1$,

\begin{equation}\label{eq:Sbound}
    S_\ell\;\ge\;\Bigl(\prod_{s\in I_\ell}\sigma_s\Bigr)\sum_{s\in I_\ell}\alpha_s
    \;\ge\;\prod_{s\in I_\ell}\sigma_s ,
    \qquad
    \prod_{s\in I_m}e_{qq,s}\;\ge\;\prod_{s\in I_m}\sigma_s ,
\end{equation}

using $\sigma_s\le1$ to obtain the missing factor $\sigma_{s}$ at the jump time. We next obtain a lower bound on $\prod_{s\in I}\sigma_s$, for any of the sub-blocks $I=I_1,\dots,I_{n-1}$. To do so, we split $I$ into small steps $\{s:K\lambda_s\le\frac12\}$ and large steps $\{s:K\lambda_s>\frac12\}$.

\begin{itemize}
    \item On a small step, $\sigma_s\ge1-K\lambda_s\ge e^{-2K\lambda_s}$ (since $1-x\ge e^{-2x}$ for $x\in[0,\tfrac12]$), so that, using $\lambda_s=\alpha_s/(C^2\aund)$ and $\sum_{I}\alpha_s<2$,
    
    \begin{equation*}
        \prod_{\text{small}}\sigma_s
        \;\ge\;\exp\Bigl(-2K\sum_{s\in I}\lambda_s\Bigr)
        \;\ge\;\exp\Bigl(-\frac{4K}{C^2\,\aund}\Bigr)
        =\exp\Bigl(-\frac{4\abar}{C^4\,\aund}\Bigr)=:h_0>0 .
    \end{equation*}
    
    \item A large step has $\lambda_s>1/(2K)$, so it carries $\alpha_s=C^2\aund\lambda_s>C^2\aund/(2K)=C^4\aund/(2\abar)=:\mu_0$. Since $\sum_{I}\alpha_s<2$, there are necessarily fewer than $2/\mu_0=4\abar/(C^4\aund)=:N_0$ of them, and $\sigma_s\ge h_1$ on each, so $\prod_{\text{large}}\sigma_s\ge h_1^{N_0}$ (as $h_1\le1$).
\end{itemize}

Thus $\prod_{s\in I}\sigma_s\ge h_0h_1^{N_0}=:h_\star\in(0,1]$, a constant depending only on $(C,\aund,\abar)$ (recall that $C>0$ $\mathbb{P}^{\tstar}$-almost surely by Proposition~\ref{prop:smg}, and is defined on the sequence of signals), and by \eqref{eq:Sbound} each of the $n-1$ factors in \eqref{eq:Sfactor} is at least $h_\star$. Therefore

\begin{equation*}
    Q_i\;\ge\;h_\star^{\,n-1}=:p_\star>0 ,
\end{equation*}

uniformly in $i\in D^c$ and in the window. Since $W$ is the same for every row, \eqref{eq:Qsplit} gives $\bigl\lVert\prod_{s\in W}P_s\bigr\rVert_\infty=\max_{i\in D^c}(1-Q_i)\le1-p_\star$, and therefore $\bigl\lVert\prod_{s\in J}P_s\bigr\rVert_\infty\le1-p_\star$. Note that $p_\star$ is a random variable, this argument is used pathwise, $p_\star>0$ $\mathbb{P}^{\tstar}$-almost surely, depending on $C>0$ $\mathbb{P}^{\tstar}$-almost surely.

\emph{Conclusion.} Group the factors of $\Phi(t,\tau)$ into the $K(t)$ complete windows it spans (any leftover factors form products having norm at most one). By submultiplicativity, $\norminf{\Phi(t,\tau)}\le(1-p_\star)^{K(t)}$, and $K(t)\to\infty$ as $t\to\infty$, so $\norminf{\Phi(t,\tau)}\to0$.
\end{proof}

We can now prove Theorem~\ref{thm:main}.

\begin{proof}[Proof of Theorem~\ref{thm:main}]
Suppose first that $D^c$ is non empty. Iterating the recurrence of Lemma~\ref{lem:recursion} and using Lemma~\ref{lem:telescope}, one has, for all $\tau$ and all $t\ge\tau$,

\begin{equation*}
    u_t\;\le\;\Phi(t,\tau)\,u_\tau+\sum_{s=\tau}^{t-1}\Phi(t,s+1)\,\beta_s\,g_s
    \;\le\;\Phi(t,\tau)\,u_\tau+\Bigl(\sup_{s\ge\tau}\beta_s\Bigr)\bone .
\end{equation*}

Fix $\epsilon>0$. Since $\beta_s$ converges to $0$ $\mathbb{P}^{\tstar}$-almost surely by Proposition~\ref{prop:vanish} together with Proposition~\ref{prop:smg} and the Bayes' rule, it implies than $r_{i,t}$ converges to 0 for all $i\in D$. Consider $\tau_{\epsilon}$ such that $\sup_{s\ge\tau_{\epsilon}}\beta_s\le\epsilon$. Moreover, since $C>0$ $\mathbb{P}^{\tstar}$-almost surely by Proposition~\ref{prop:smg}, we also know from Lemma~\ref{lem:contraction} that, if $D^c$ is not empty,

\begin{equation*}
    \norminf{\Phi(t,\tau)u_\tau}\le\frac{\norminf{\Phi(t,\tau)}}{C}\xrightarrow[t\to\infty]{}0 .
\end{equation*}

Therefore $\limsup_t\norminf{u_t}\le\epsilon$; since this holds for all $\epsilon>0$, it implies that

\begin{equation*}
    \lim_t\norminf{u_t}=0,\qquad\text{and then}\qquad \lim_t u_t=0 .
\end{equation*}

This proves the Theorem, because, $\forall i\in D^c$ $0<r_{i,t}\le\norminf{u_t}$ with $r_{i,t}=\frac{\mu_{it}(\that)}{\mu_{it}(\tstar)}$ and $\mu_{i,t}(\tstar)\ge C$ by Proposition~\ref{prop:smg}, the beliefs of all agents for this false state converge to 0 $\mathbb{P}^{\tstar}$-almost surely. combined with Proposition \ref{prop:vanish} this is equivalent to Theorem~\ref{thm:main}, because $\sum_{\theta\in\Theta}\mu_{i,t}(\theta)=1$ by the definition of beliefs. If $D^c$ is empty, then all agents are capable of distinguishing the true state from $\that$ and we can therefore conclude by Proposition \ref{prop:vanish} 
\end{proof}

\section{Appendix}\label{app:two}

\subsection{Proof of Theorem~\ref{thm:linearrep}}\label{app:T2}

\begin{proof}[Proof of Theorem~\ref{thm:linearrep}]
(i) Fix a state $\that\ne\tstar$. Define $F_{i,t}(u)=\log\psi_{i,t}(e^u)$, which is $C^1$ since $\psi_{i,t}$ is $C^2$ along the segment $\{u(s):=(1-s)\log\mu_t(\that)+s\log\mu_t(\tstar):s\in[0,1]\}$ of Definition~\ref{def:logweights}. Therefore, by \eqref{eq:interim} and \eqref{eq:bayes},

\begin{equation*}
    x_{i,t+1}(\that)
    =\log\frac{\zeta_{it}(\tstar)}{\zeta_{it}(\that)}+y_{i,t+1}(\that)
    =\Bigl[F_{i,t}\bigl(\log\mu_t(\tstar)\bigr)-F_{i,t}\bigl(\log\mu_t(\that)\bigr)\Bigr]+y_{i,t+1}(\that).
\end{equation*}

We can now use the fundamental Theorem of calculus to obtain the integral form

\begin{equation*}
    F_{i,t}\bigl(u(1)\bigr)-F_{i,t}\bigl(u(0)\bigr)
    =\int_0^1\nabla F_{i,t}\bigl(u(s)\bigr)\,ds\cdot\bigl(u(1)-u(0)\bigr)
    =\sum_{j=1}^n\tilde e_{ij,t}\,x_{jt},
\end{equation*}

where the first equality comes from the chain rule $\frac{d}{ds}u(s)=u(1)-u(0)$ and the second equality uses Definition~\ref{def:logweights}. Indeed, $\partial F_{i,t}/\partial u_j$ evaluated at $u$ equals $e_{ij,t}(e^u)$ by Definition~\ref{def:curv}, and $u_j(1)-u_j(0)=x_{jt}(\that)$.

(ii) The monotonicity of $\psi_{i,t}$ assumed in Hypothesis~\ref{hyp:H1} gives $e_{ij,t}\ge0$, and Euler's identity for the degree-one homogeneous function $\psi_{i,t}$ gives $\sum_jx_j\partial_{x_j}\psi_{i,t}=\psi_{i,t}$, that is, $\sum_je_{ij,t}(x)=1$ at every $x$, as recorded in Lemma~\ref{lem:elast}; integrating over $s$ preserves both properties. Measurability holds because $u(\cdot)$ is a deterministic function of $\mu_t$.
\end{proof}

\subsection{Proof of Theorem~\ref{thm:upperbound}}\label{app:T3}.

\begin{proof}[Proof of Theorem~\ref{thm:upperbound}]
Iterating part (i) of Theorem~\ref{thm:linearrep}, we obtain, for all $\that\ne\tstar$,

\begin{equation}\label{eq:variation}
    x_t(\that)=\tilde{\Phi}(t,0)\,x_0(\that)+\sum_{\tau=1}^{t}\tilde{\Phi}(t,\tau)\,y_\tau(\that),
    \qquad
    \tilde{\Phi}(t,\tau):=\tilde{E}_{t-1}(\that)\cdots\tilde{E}_\tau(\that) .
\end{equation}

A product of row-stochastic matrices is itself row-stochastic, by part (ii) of Theorem~\ref{thm:linearrep}, therefore $\tilde\Phi(t,\tau)$ is row-stochastic, and the first term on the right-hand side is  bounded (in absolute value) by some constant $\tilde{c}$. Moreover,  $[\tilde{\Phi}(t,\tau)y_\tau(\that)]_i\le\max_jy_{j\tau}(\that)$, again because $\tilde\Phi(t,\tau)$ is row-stochastic. This leads to $x_{it}(\that)\le \tilde{c}+\sum_{\tau=1}^t\xi_\tau$ with $\xi_\tau:=\max_jy_{j\tau}(\that)$. The signals are independent and identically distributed by Hypothesis~\ref{hyp:H3}; thus, by the strong law of large numbers, $\frac1t\sum_{\tau\le t}\xi_\tau\to\bar h(\that)$, where $\bar h$ is as in Definition~\ref{def:max_inform}. Hence

\begin{equation*}
    \limsup_{t\to\infty}\frac{x_{it}(\that)}{t}\;\le\;\bar h(\that).
\end{equation*}

Finally, since $|\log\mu_{it}(\that)|\le |x_{it}(\that)|+|\log\mu_{it}(\tstar)|$, and by Theorem~\ref{thm:main} the second term on the right-hand side converges to $0$, and the first term will be eventually positive $\mathbb{P}^{\tstar}$-almost surely, we therefore reach the conclusion.
\end{proof}

\subsection{Sufficient hypotheses}\label{app:H}

Before proving Theorems \ref{thm:lower} and \ref{thm:upper} and Proposition \ref{prop:S=R}, we present a set of constraints which are sufficient to prove these Theorems and including Hypotheses \ref{hyp:H1} to \ref{hyp:H4} and including.

\begin{hypothesis}[Rate Hypothesis]\label{hyp:R}
Consists of Hypotheses \ref{hyp:H1} to \ref{hyp:H4}, from which we remove the assumption of independence across agents from Hypothesis \ref{hyp:H3} and we replace inequality \ref{eq:constraint} by the concavity of $\psi_{i,t}$ for all $i\in \mathcal{N}$. Moreover, we assume $\lim_{t\to\infty}\mu_{it}(\theta^*)=1$ $\Pstar\text{-almost surely, for every } i\in\cN$
\end{hypothesis}

\begin{remark}
    Because of Theorem $\ref{thm:main}$ and Lemma A.9 from \cite{molavi2018theory}, Hypothesis $\ref{hyp:R}$ is implied by hypotheses \ref{hyp:H1} to \ref{hyp:H4}. This makes the results of this Section stronger. In particular, we only need convergence on the true state and we remove the independence of signals across agents assumption. Moreover, it would actually still works if the aggregators were themselves stochastic.
\end{remark}

\subsection{Proof of Theorem~\ref{thm:lower}}\label{app:T4}.

\begin{proof}[Proof of Theorem~\ref{thm:lower}]
Fix $c>0$, $\that$ and an agent $i\in\mathcal{N}$ let $j^\star\in R_c(i)$, with $h_j$ as in Definition~\ref{def:max_inform}. Let $r_{it}:=\frac{\mu_{it}(\that)}{\mu_{it}(\tstar)}\ge0$. Then, by \eqref{eq:interim} and \eqref{eq:bayes},

\begin{equation*}
    r_{i,t+1}=L_{i,t+1}\,\frac{\psi_{i,t}\bigl(\mu_t(\that)\bigr)}{\psi_{i,t}\bigl(\mu_t(\tstar)\bigr)},
    \qquad
    L_{i,t+1}:=\frac{\ell^{\that}_i(\omega_{i,t+1})}{\ell^{\tstar}_i(\omega_{i,t+1})} .
\end{equation*}

Applying the definition of $R_c(i)$ to $\psi_{j^\star,t}$ with $j=j^\star$ (since $j^*$ is self dominating), and using Hypothesis~\ref{hyp:H2}, we obtain

\begin{equation*}
    \psi_{j^\star,t}\bigl(\mu_t(\that)\bigr)\le\bigl(1-\bar a\lambda_t\bigr)^{-1/c}\mu_{j^\star t}(\that)
    \;=\;e^{\,\frac{1}{c}|\log(1-\bar a\lambda_t)|}\,\mu_{j^\star t}(\that).
\end{equation*}

Moreover $\mu_{j^\star t}(\that)\le r_{j^\star,t}$, so

\begin{equation*}
    r_{j^\star,t+1}\;\le\;\frac{e^{\frac{1}{c}|\log(1-\bar a\lambda_t)|}}{\psi_{j^\star,t}(\mu_t(\tstar))}\,L_{j^\star,t+1}\,r_{j^\star,t}.
\end{equation*}

Taking logarithms on both sides, one obtains

\begin{equation*}
    \log r_{j^\star,t}\le\log r_{j^\star,0}
    +\sum_{\tau=1}^{t}\log L_{j^\star,\tau}
    +\frac{1}{c}\sum_{\tau=0}^{t-1}\bigl|\log(1-\bar a\lambda_\tau)\bigr|
    +\sum_{\tau=0}^{t-1}\log\frac1{\psi_{j^\star,\tau}(\mu_\tau(\tstar))}.
\end{equation*}

We can then divide by $t$ and use the strong law of large numbers, together with the definition of $j^\star$ and Cesàro's Theorem: by Hypothesis \ref{hyp:R}, $\mu_t(\tstar)$ converges to $\bone$, and $\psi_{j^\star,t}(\mu_t(\tstar))$ is lower bounded by the minimum of the elements of $\mu_t(\theta^*)$ therefore, $\psi_{j^\star,t}(\mu_t(\tstar))$ also converges to $1$. The result then follows, $\mathbb{P}^{\tstar}$-almost surely, by taking the limit inferior (note the sign flip due to the eventual negativity of $\log(r_{i,t})$ due to the convergence of $\mu_{i,t}(\tstar)$ to 1):

\begin{equation*}
    \liminf_{t\to\infty}\frac1t\bigl|\log r_{j^\star t}(\that)\bigr|\;\ge\;h_{j^{\star}}(\that)-\frac{\beta_\lambda}{c}.
\end{equation*}

We now use the strong connectivity of Hypothesis~\ref{hyp:H2} to propagate this result along the communication graph. Fix $i\neq j$ and a directed dominating path $i=q_0\to q_1\to\dots\to q_L=j^\star$. Using the definition of such a path, since $(q_l,q_{l+1})$ is $c$-dominating, one obtains

\begin{equation*}
    r_{q_\ell,t+1}\le L_{q_\ell,t+1}\,\frac{\bigl(\underline a\lambda_t\bigr)^{-1/c}}{\psi_{q_\ell,t}(\mu_t(\tstar))}\,r_{q_{\ell+1},t} .
\end{equation*}

Taking logarithms and using Hypothesis~\ref{hyp:R}, we get

\begin{equation*}
    \log(r_{q_\ell,t+1})\le\log(L_{q_\ell,t+1})-\log(\psi_{q_\ell,t}(\mu_t(\tstar)))+\frac{1}{c}\log\Bigl(\frac{1}{\underline a \lambda_t}\Bigr)+\log(r_{q_{\ell+1},t}).
\end{equation*}

Since $t\lambda_t$ tends to $+\infty$ by Hypothesis~\ref{hyp:H2}, there exists a time $T$ such that $\lambda_t\ge\frac{1}{t}$ for all $t\ge T$. Dividing the previous inequality by $t$, we therefore obtain

\begin{equation*}
    \liminf_{t\to\infty}\frac1t\bigl|\log r_{q_\ell,t}(\that)\bigr|\;\ge\;\liminf_{t\to\infty}\frac1t\bigl|\log r_{q_{\ell+1},t}(\that)\bigr| .
\end{equation*}

We prove the Theorem by iterating along the path, using the definition of $r_{i,t}$ and the fact that $\frac{\log(\mu_{i,t}(\tstar))}{t}$ converges to $0$ $\mathbb{P}^{\tstar}$-almost surely, which follows from the convergence of $\mu_{i,t}(\tstar)$ to 1. We can therefore finish the proof, taking the maximum over $j^*\in R_c(i)$ and the supremum over $c>0$.
\end{proof}

\subsection{Proof of Theorem~\ref{thm:upper}}\label{app:T5}.

\begin{proof}[proof of Theorem~\ref{thm:upper}].
Fix $c>0$ and $j\in S_c(i)$ and a $c$-dominated path $i=q_0\to\cdots\to q_L=j$. By definition, at $x=\mu_t(\that)$, together with $a_{jj,t+1}\ge1-\abar\lambda_t$ and $a\mapsto a^{1/c}$ increasing, $\psi_{j,t}(\mu_t(\that))\ge e^{-\frac1c|\log(1-\abar\lambda_t)|}\mu_{jt}(\that)$. The Bayesian update gives $\mu_{j,t+1}(\that)=\zeta_{jt}(\that)\ell^{\that}_j(\omega_{j,t+1})/m_{jt}(\omega_{j,t+1})$, and equation \ref{eq:interim} together with $S_{jt}\le1$ (obtained from the concavity of $\psi_{j,t}$) gives $\zeta_{jt}(\that)=\psi_{j,t}(\mu_t(\that))/S_{jt}\ge\psi_{j,t}(\mu_t(\that))$. Hence

\begin{equation*}
    \mu_{j,t+1}(\that)\;\ge\;e^{-\frac1c|\log(1-\abar\lambda_t)|}\,
    \frac{\ell^{\that}_j(\omega_{j,t+1})}{m_{jt}(\omega_{j,t+1})}\,\mu_{jt}(\that).
\end{equation*}

Taking logarithms, iterating, and splitting $\ell^{\that}_j/m_{jt}=(\ell^{\that}_j/\ell^{\tstar}_j)\cdot(\ell^{\tstar}_j/m_{jt})$,

\begin{equation*}
    \frac1t\log\mu_{jt}(\that)\ge\frac{\log\mu_{j0}(\that)}{t}
    -\frac1{ct}\sum_{\tau<t}\bigl|\log(1-\abar\lambda_\tau)\bigr|
    +\frac1t\sum_{\tau\le t}\log\frac{\ell^{\that}_j(\omega_{j\tau})}{\ell^{\tstar}_j(\omega_{j\tau})}
    +\frac1t\sum_{\tau\le t}\log\frac{\ell^{\tstar}_j(\omega_{j\tau})}{m_{j,\tau-1}(\omega_{j\tau})} .
\end{equation*}

The third term tends to $-h_j(\that)$ by the strong law. The fourth vanishes: by the convergence of $\mu_{k,t}(\tstar)$ to 1 for all $k\in\mathcal{N}$ and 

\begin{equation*}
    \min_{k\in\cN}\mu_{kt}(\tstar)\;\le\;\psi_{j,t}\bigl(\mu_t(\tstar)\bigr)\;\le\zeta_{j,t}(\tstar)\le\;1 .
\end{equation*}

$\zeta_{jt}(\tstar)\to1$, hence $m_{jt}(\omega)=\sum_\theta\zeta_{jt}(\theta)\ell^\theta_j(\omega)\to\ell^{\tstar}_j(\omega)$ for each $\omega$; the sum is finite and the likelihoods lie in $[c_\ell,1]$, so the logarithms are bounded and Ces\`aro applies. Hence (note the sign flip due to the eventual negativity of $log(r_{i,t})$ due to the convergence of $\mu_{i,t}(\tstar)$ to 1)

\begin{equation}\label{eq:ustep1}
  \limsup_{t\to\infty}\frac1t\bigl|\log\mu_{jt}(\that)\bigr|\le h_j(\that)+\frac{\beta_\lambda}{c}.
\end{equation}

By definition applied to $\psi_{q_\ell,t}$ with $q_{\ell+1}$ and $a_{q_\ell q_{\ell+1},t+1}\ge\aund\lambda_t$, then Bayes and $S_{q_\ell t}\le1$ as in the first part of proof,

\begin{equation*}
    \mu_{q_\ell,t+1}(\that)\;\ge\;(\aund\lambda_t)^{1/c}\,
    \frac{\ell^{\that}_{q_\ell}(\omega_{q_\ell,t+1})}{m_{q_\ell t}(\omega_{q_\ell,t+1})}\,
    \mu_{q_{\ell+1}t}(\that).
\end{equation*}

Taking logarithms and dividing by $t$: the likelihood factor is a single term bounded by $\log(1/c_\ell)$ in absolute value, hence $O(1/t)$; and $\frac1{ct}\log\frac1{\aund\lambda_t}\to0$ as before. Hence

\begin{equation}\label{eq:ustep2}
  \limsup_{t\to\infty}\frac1t\bigl|\log\mu_{q_\ell t}(\that)\bigr|
  \le\limsup_{t\to\infty}\frac1t\bigl|\log\mu_{q_{\ell+1}t}(\that)\bigr| .
\end{equation}

Iterating \eqref{eq:ustep2} along the path and applying \eqref{eq:ustep1} at $j$ gives the bound with $h_j(\that)$; taking the minimum over $j\in S_c(i)$ and the infimum over $c>0$ completes the proof.
\end{proof}

\subsection{Proof of Proposition \ref{prop:S=R}}\label{app:P1}

\begin{proof}[proof of Proposition~\ref{prop:S=R}]
\textbf{Point (i)} Suppose there exist $c_1>0$ and $c_2>0$ such that, for an agent $i\in\mathcal{N}$ if, for some $j$ and $k$, different neighbors of $i$:

\begin{equation*}
    a_{ik,t+1}^{\;-1/{c_1}}\;x_k\ge \psi_{i,t}(x)\;\ge\;a_{ij,t+1}^{\;1/{c_2}}\;x_j \qquad\text{for every } x\in(0,\infty)^n.
\end{equation*}

For example, taking $x_j=1$ and $x_k$ converging to 0 gives a contradiction. Therefore, if $\psi_{i,t}$ is $c_1$-dominating $n_1\in\mathcal{N}_i$, and $c_2-$dominated in $n_2\in\mathcal{N}_i$ neighbors, then $n_1=n_2$.

Now, let $i=p_0\to\cdots\to p_M=r_1$ a $c_1$-dominating path with $r_1\in R_{c_1}(i)$ and $i=q_0\to\cdots\to q_N=r_2$ be a $c_2$-dominated path with $r_2\in S_{c_2}(i)$ (remember that we suppose both $R_{c_1}(i)\neq\emptyset$ and $S_{c_2}(i)\neq\emptyset$). We consider that these are the smallest possible paths (i.e not repeating the final loops, the only $c_1$-self dominating or $c_2$-self dominated agent is the last). Then, by induction, we show $N=M$ and $p_l=q_l$ for $0\le l\le N$. The base case is direct since $p_0=q_0=i$. The induction step relies on the preceding point, if $p_l=q_l$ for some $0\le l\le \min(N,M)-1$, for all $t\ge0$, $\psi_{{p_l},t}=\psi_{{q_l},t}$ is $c_1$-dominating in $p_{l+1}$ and $c_2$-dominated in $q_{l+1}$ therefore $p_{l+1}=q_{l+1}$. Moreover, both paths have to finish in the same agent, let us suppose $M\le N$, $p_M=q_M$, since the $c_1$-dominating path terminates at $p_M$, it means that this agent is $c_1$-self-dominating, and $q_M$ is therefore $c_2$-self-dominated, using the preceding argument, and the $c_2$-dominated path terminates too and $N=M$. Symmetrically, the same argument holds if we suppose $N\le M$.

Thus, $p_N=q_M$ implies $r_1=r_2$, $r_1$ is equal to all the elements of $S_{c_2}(i)$ and $r_2$ is equal to all the elements of $R_{c_1}(i)$, therefore, these sets are equal and contain one element $S_{c_2}(i)=R_{c_1}(i)=\{r\}$ with $r=r_1=r_2$. \textbf{Point (ii)} This is a direct corollary of (i) combining with Theorems \ref{thm:lower} and \ref{thm:upper}.
\end{proof}

\subsection{Proof of Proposition \ref{prop:curv}}\label{app:P2}

\begin{proof}[proof of Proposition~\ref{prop:curv}] Fix $i$, an agent of the graph and $j$ one of her neighbors, satisfying the Hypothesis of the Proposition with a sequence $\psi_{it}$ of aggregation functions verifying H1.

(i): Define $x'$ the vector of $(0,\infty)^n$ such that $x'_l=\max(x_j,x_l)$. Moreover, define the continuous line connecting $x_j\mathbf{1}$ to $x'$:

\begin{equation}\label{eq:path}
  x(s):=\Bigl(x_j\,\bigl(x'_l/x_j\bigr)^{s}\Bigr)_{l=1}^{n},\qquad s\in[0,1],
\end{equation}

note that $x(0)=x_j\bone$ and $x(1)=x'$. Moreover, since $x_j\mathbf{1}\le x'$, the monotonicity of $\psi_{it}$ implies $\psi_{i,t}(x)\le\psi_{i,t}(x')$.

Also define $d_l:=\log\frac{x'_l}{x_j}\ge0$, one has $\frac{d}{ds}\log x_l(s)=d_l$, and, using the chain rule and  definition~\ref{def:curv},

\begin{equation}\label{eq:Ddef}
  \frac{d}{ds}\log\psi_{i,t}\bigl(x(s)\bigr)
  =\sum_l\frac{\partial\log\psi_{i,t}}{\partial\log x_l}\bigl(x(s)\bigr)\,d_l
  =\sum_le_{il,t}\bigl(x(s)\bigr)\,d_l\;=:\;D(s)\;\ge\;0 .
\end{equation}

Therefore, 

\begin{equation}\label{eq:FD}
  \log\psi_{i,t}(x')-\log x_j=\int_0^1D(s)\,ds .
\end{equation}

Moreover, using the chain rule again and $d_j=0$, one has,

\begin{equation*}
    \frac{d}{ds}e_{ij,t}\bigl(x(s)\bigr)
    =\sum_l\frac{\partial e_{ij,t}}{\partial\log x_l}\bigl(x(s)\bigr)\,d_l
    =\sum_lH^{i,t}_{jl}\bigl(x(s)\bigr)\,d_l
    =\sum_{l\neq j}H^{i,t}_{jl}\bigl(x(s)\bigr)\,d_l .
\end{equation*}

The Hypothesis of the Proposition then gives immediately 

\begin{equation*}
    \frac{d}{ds}e_{ij,t}\bigl(x(s)\bigr)
    \;\ge\;c\,e_{ij,t}\bigl(x(s)\bigr)\sum_{l\neq j}e_{il,t}\bigl(x(s)\bigr)d_l
    \;=\;c\,e_{ij,t}\bigl(x(s)\bigr)\,D(s).
\end{equation*} 

We can now divide by $e_{i,j,t}>0$ since $j$ is a neighbor of $i$, and integrate between $0$ and $1$, and note that $e_{ij,t}$ is homogeneous of degree 0, so $e_{ij,t}\bigl(x_j\bone\bigr)=e_{ij,t}\bigl(\bone\bigr)=a_{ij,t+1}$, resulting in the following identity,
\begin{equation*}
    c\int_0^1D(s)\,ds\;\le\;\log e_{ij,t}(x')-\log e_{ij,t}\bigl(x_j\bone\bigr)
    \;=\;\log e_{ij,t}(x')-\log a_{ij,t+1}\;\le\;\bigl|\log a_{ij,t+1}\bigr| .
\end{equation*}

Combining with \eqref{eq:FD}, and Lemma~\ref{lem:elast} (implying $\log(e_{ij,t}(x'))\le0$)

\begin{equation*}
    \log\psi_{i,t}(x)\;\le\;\log\psi_{i,t}(x')\;\le\;\log x_j+\frac1c\bigl|\log a_{ij,t+1}\bigr| ,
\end{equation*}

which is \eqref{eq:dominating} after exponentiating. 

(ii): We simply have to apply (i) to the dual functions of the sequence $(\psi_{i,t})_t$. We know from lemma \eqref{lem:phi} that $\phi_{i,t}(x):=\bigl(\psi_{i,t}(1/x_1,\dots,1/x_n)\bigr)^{-1}$ also satisfies Hypothesis 1 and has the same weights as $\psi_{i,t}$ and 

\begin{equation*}
    e^{\phi(i)}_{ij,t}(x)=e_{ij,t}(1/x),
    \qquad
    H^{\phi(i),t}_{jl}(x)=-H^{i,t}_{jl}(1/x) .
\end{equation*}

Therefore, applying (i) to this dual sequence gives $\phi_{i,t}(y)\le a_{ij,t+1}^{-1/c}y_j$ for every $y>0$ with $y=1/x$, giving $\psi_{i,t}(x)\ge a_{ij,t+1}^{1/c}x_j$ finishing the proof.
\end{proof}

\newpage
\bibliographystyle{abbrvnat}
\bibliography{bibliography}

\end{document}